\documentclass[preprint,review,12pt]{elsarticle_arxiv}

\usepackage{graphicx}
\usepackage{amsfonts,amstext,amsmath,amssymb,amsthm}

\usepackage{txfonts}

\usepackage{float}

\usepackage{subfig}

\usepackage{tabularx, booktabs}
\newcolumntype{Y}{>{\centering\arraybackslash}X}

\usepackage{empheq}
\usepackage{paralist}

\usepackage[%
    bookmarks=true,         % show bookmarks bar?
    unicode=false,          % non-Latin characters in Acrobat’s bookmarks
    pdftoolbar=true,        % show Acrobat’s toolbar?
    pdfmenubar=true,        % show Acrobat’s menu?
    pdffitwindow=false,     % window fit to page when opened
    pdfstartview={FitH},    % fits the width of the page to the window
    pdftitle={Efficient Performance Evaluation of the Generalized Shiryaev--Roberts Detection Procedure in a Multi-Cyclic Setup},    % title
    pdfauthor={Aleksey S. Polunchenko and Grigory Sokolov and Wenyu Du},     % author
    pdfsubject={Efficient Performance Evaluation of the Generalized Shiryaev--Roberts Detection Procedure in a Multi-Cyclic Setup},   % subject of the document
    pdfcreator={Aleksey S. Polunchenko},   % creator of the document
    pdfproducer={MikTeX}, % producer of the document
    pdfkeywords={Sequential Analysis}{Quickest change-point detection}{Shiryaev--Roberts procedure}{Generalized Shiryaev--Roberts procedure}{Sequential analysis}, % list of keywords
    pdfnewwindow=true,      % links in new window
    colorlinks=false,       % false: boxed links; true: colored links
    linkcolor=red,          % color of internal links (change box color with linkbordercolor)
    citecolor=green,        % color of links to bibliography
    filecolor=magenta,      % color of file links
    urlcolor=cyan           % color of external links
]{hyperref}

\usepackage[T1]{fontenc}

\biboptions{numbers,compress}

\usepackage{multirow}

\renewcommand{\Pr}{\mathbb{P}} % probability
\DeclareMathOperator{\EV}{\mathbb{E}} % expected value

\DeclareMathOperator{\LR}{\Lambda}

\DeclareMathOperator{\ARL}{ARL}

\DeclareMathOperator{\IADD}{IADD}
\DeclareMathOperator{\STADD}{STADD}
\DeclareMathOperator{\RIADD}{RIADD}
\DeclareMathOperator{\SSADD}{SSADD}

\newcommand{\T}{T}

\renewcommand{\le}{\leqslant} % AMS le ge
\renewcommand{\ge}{\geqslant}

\newcommand{\abs}[1]{\left\vert#1\right\vert}

\DeclareMathOperator{\One}{\mathchoice{\rm 1\mskip-4.2mu l}{\rm 1\mskip-4.2mu l}{\rm 1\mskip-4.6mu l}{\rm 1\mskip-5.2mu l}}
\newcommand{\indicator}[1]{{\One_{\left\{#1\right\}}}}

\theoremstyle{plain} %% produces italic text
\newtheorem{theorem}{Theorem}
\newtheorem{lemma}{Lemma}

\theoremstyle{remark}
\newtheorem*{remark}{Remark}

\newcommand{\ignore}[1]{}

\graphicspath{{./gfx/}}

\journal{Applied Stochastic Models in Business and Industry}

\begin{document}

%-------------------------------------------------------------------------------------------------%
\begin{frontmatter}



\title{\large{\bf\uppercase{Efficient Performance Evaluation of the Generalized Shiryaev--Roberts Detection Procedure in a Multi-Cyclic Setup}}}


\author[bu]{Aleksey\ S.\ Polunchenko\corref{cor-author}}
\ead{aleksey@binghamton.edu}
\ead[url]{http://www.math.binghamton.edu/aleksey}
\author[usc]{Grigory\ Sokolov}
\ead{gsokolov@usc.edu}
\ead[url]{http://cams.usc.edu/~gsokolov}
\author[bu]{Wenyu\ Du}
\ead{wdu@binghamton.edu}
\ead[url]{http://www.math.binghamton.edu/grads/wdu}
\cortext[cor-author]{Address correspondence to A.S.\ Polunchenko, Department of Mathematical Sciences, State University of New York (SUNY) at Binghamton, Binghamton, NY 13902--6000, USA; Tel: +1 (607) 777-6906; Fax: +1 (607) 777-2450; Email:~\href{mailto:aleksey@binghamton.edu}{aleksey@binghamton.edu}}
\address[bu]{Department of Mathematical Sciences, State University of New York (SUNY) at Binghamton\\Binghamtom, NY 13902--6000, USA}
\address[usc]{Department of Mathematics, University of Southern California\\Los Angeles, CA 90089--2532, USA}

\begin{abstract}

% ASMBI requirement: no more than 250 words
We propose a numerical method to evaluate the performance of the emerging Generalized Shiryaev--Roberts (GSR) change-point detection procedure in a ``minimax-ish'' multi-cyclic setup where the procedure of choice is applied repetitively (cyclically) and the change is assumed to take place at an unknown time moment in a distant-future stationary regime. Specifically, the proposed method is based on the integral-equations approach and uses the collocation technique with the basis functions chosen so as to exploit a certain change-of-measure identity and the GSR detection statistic's unique martingale property. As a result, the method's accuracy and robustness improve, as does its efficiency since using the change-of-measure ploy the Average Run Length (ARL) to false alarm and the Stationary Average Detection Delay (STADD) are computed {\em simultaneously}. We show that the method's rate of convergence is quadratic and supply a tight upperbound on its error. We conclude with a case study and confirm experimentally that the proposed method's accuracy and rate of convergence are robust with respect to three factors:\begin{inparaenum}[\itshape a)]\item partition fineness (coarse vs. fine), \item change magnitude (faint vs. contrast), and \item the level of the ARL to false alarm (low vs. high)\end{inparaenum}. Since the method is designed not restricted to a particular data distribution or to a specific value of the GSR detection statistic's headstart, this work may help gain greater insight into the characteristics of the GSR procedure and aid a practitioner to design the GSR procedure as needed while fully utilizing its potential.
\end{abstract}


\end{frontmatter}
%-------------------------------------------------------------------------------------------------%

%-------------------------------------------------------------------------------------------------%
% SUBJECT CLASSIFICATIONS
%
% MSC2010, see http://www.ams.org/msc/
%
%  Statistics -> Sequential methods
%   62L10 - Sequential analysis
%   62L15 - Optimal stopping
%
%  Statistics -> Applications
%   62P30 - Applications in engineering and industry
%
%{\small\noindent\textbf{Subject Classifications:} 62L10; 62L15; 62P30.}

%+-----------------------------------------------------------------------------------------------+%
\section{Introduction}
\label{sec:intro}

Sequential (quickest) change-point detection is concerned with the design and analysis of statistical machinery for ``on-the-go'' detection of unanticipated changes that may occur in the characteristics of a ongoing (random) process. Specifically, the process is assumed to be continuously monitored through sequentially made observations (e.g., measurements), and should their behavior suggest the process may have statistically changed, the aim is to conclude so within the fewest observations possible, subject to a tolerable level of the false detection risk~\cite{Shiryaev:Book78,Basseville+Nikiforov:Book93,Poor+Hadjiliadis:Book08}. The subject's areas of application are diverse and virtually unlimited, and include industrial quality and process control~\cite{Shewhart:JASA1925,Shewhart:Book1931,Kenett+Zacks:Book1998,Ryan:Book2011,Montgomery:Book2012}, biostatistics, economics, seismology~\cite{Basseville+Nikiforov:Book93}, forensics, navigation~\cite{Basseville+Nikiforov:Book93}, cybersecurity~\cite{Tartakovsky+etal:JSM2005,Tartakovsky+etal:SM2006-discussion,Polunchenko+etal:SA2012,Tartakovsy+etal:IEEE-JSTSP2013}, communication systems~\cite{Basseville+Nikiforov:Book93}, and many more. A sequential change-point detection procedure---a rule whereby one stops and declares that (apparently) a change is in effect---is defined as a stopping time, $\T$, that is adapted to the observed data, $\{X_n\}_{n\ge1}$.

This work's focus is on the multi-cyclic change-point detection problem. It was first addressed in~\cite{Shiryaev:SMD61,Shiryaev:TPA63} in continuous time; see also, e.g.,~\cite{Shiryaev:Bachelier2002,Feinberg+Shiryaev:SD2006}. We consider the basic discrete-time case~\cite{Pollak+Tartakovsky:SS09,Shiryaev+Zryumov:Khabanov2010} which assumes the observations are independent throughout the entire period of surveillance with the pre- and post-change distributions fully specified (but not equal to one another). The change-point is treated as an unknown (but not random) nuisance parameter and is assumed to take place in a distant-future stationary regime. That is, the process of interest is not expected to change soon, and is monitored by applying the detection procedure of choice repetitively, or cyclically (hence, the name ``multi-cyclic''), starting anew each time a false alarm (appearance of a change) is sounded. This is known~\cite{Shiryaev:SMD61,Shiryaev:TPA63,Shiryaev:Bachelier2002,Feinberg+Shiryaev:SD2006,Pollak+Tartakovsky:SS09,Shiryaev+Zryumov:Khabanov2010} to be equivalent to the generalized Bayesian change-point detection problem (see, e.g.,~\cite{Tartakovsky+Moustakides:SA10,Polunchenko+Tartakovsky:MCAP2012,Polunchenko+etal:JSM2013} for an overview), and is a reasonable approach provided the cost of a false alarm is relatively small compared to the cost of a unit of delay to reach the conclusion that the process is ``out-of-control'' {\em post-change}. Such scenarios occur, e.g., in cybersecurity~\cite{Polunchenko+etal:SA2012,Tartakovsy+etal:IEEE-JSTSP2013} and in the economic design of quality control charts~\cite{Duncan:JASA1956,Montgomery:JQT1980,Lorenzen+Vance:T1986,Ho+Case:JQT1994}.

Within the multi-cyclic setup, particular emphasis in the paper is placed on two related detection procedures:\begin{inparaenum}[\itshape a)]\item the original Shiryaev--Roberts (SR) procedure (due to the independent work of Shiryaev~\cite{Shiryaev:SMD61,Shiryaev:TPA63} and that of Roberts~\cite{Roberts:T66}; see also~\cite{Girschick+Rubin:AMS52}), and \item its recent generalization---the Shiryaev--Roberts--$r$ (SR--$r$) procedure introduced in~\cite{Moustakides+etal:SS11} as a version of the original SR procedure with a headstart (the ``$r$'' in the name ``SR--$r$'' is the headstart), akin to~\cite{Lucas+Crosier:T1982}\end{inparaenum}. As the SR procedure is a special case of the SR--$r$ procedure (with no headstart, i.e., when $r=0$), we will collectively refer to both as the Generalized SR (GSR) procedure, in analogy to the terminology used in~\cite{Tartakovsky+etal:TPA2012}.

Our interest in the GSR procedure is due to three reasons. First, the GSR procedure is relatively ``young'' (the SR--$r$ procedure was proposed in~\cite{Moustakides+etal:SS11} in 2011), and has not yet been fully explored in the literature. Second, in spite of the ``young age'', the GSR procedure has already been proven to be {\em exactly} multi-cyclic optimal. This was first established in~\cite{Shiryaev:SMD61,Shiryaev:TPA63} in continuous time for the problem of detecting a shift in the drift of a Brownian motion; see also, e.g.,~\cite{Shiryaev:Bachelier2002,Feinberg+Shiryaev:SD2006}. An analogous result in discrete time was later obtained in~\cite{Pollak+Tartakovsky:SS09,Shiryaev+Zryumov:Khabanov2010}, and shortly after generalized in~\cite[Lemma~1]{Polunchenko+Tartakovsky:AS10}. Neither the famous Cumulative Sum (CUSUM) ``inspection scheme''~\cite{Page:B54} nor the popular Exponentially Weighted Moving Average (EWMA) chart~\cite{Roberts:T59} possesses such strong optimality property. This notwithstanding, there is currently a vacuum in the literature concerning numerical methodology to compute the performance of the GSR procedure in the multi-cyclic setup. As a matter of fact, to the best of our knowledge, only~\cite{Moustakides+etal:CommStat09,Tartakovsky+etal:IWSM2009,Moustakides+etal:SS11} and~\cite{Polunchenko+etal:ShLnkJAS2013} address this question, and in particular, offer a comparative performance analysis of the CUSUM scheme, EWMA chart and the GSR procedure in the multi-cyclic setup; similar analysis in continuous time can be found, e.g., in~\cite{Srivastava+Wu:AS1993}. However, the question of the employed method's accuracy is only partially answered, with no error bounds or convergence rates supplied. This is common in the literature on the computational aspect of change-point detection: to deal with the accuracy question in an {\it ad hoc} manner, if even. Some headway to fill in this gap was recently made in~\cite{Polunchenko+etal:MIPT2013,Polunchenko+etal:SA2014}. The third, equally important reason to consider the GSR procedure is its asymptotic near optimality in the minimax sense of Pollak~\cite{Pollak:AS85}; see~\cite{Tartakovsky+etal:TPA2012} for the corresponding result established using the GSR procedure's exact multi-cyclic optimality. Furthermore, the GSR procedure is also proven~\cite{Tartakovsky+Polunchenko:IWAP10,Polunchenko+Tartakovsky:AS10} to be {\em exactly} Pollak-minimax optimal in two special cases (again as a consequence of the exact multi-cyclic optimality). A practical implication of this is that the CUSUM chart is less minimax efficient than the GSR procedure, and the difference is especially contrast when the change is faint; for a few particular scenarios the difference is quantified, e.g., in~\cite{Moustakides+etal:CommStat09,Tartakovsky+etal:IWSM2009,Moustakides+etal:SS11}.

To foster and facilitate further research on the GSR procedure, in this work we build on to the work done previously in~\cite{Moustakides+etal:CommStat09,Tartakovsky+etal:IWSM2009,Moustakides+etal:SS11,Polunchenko+etal:MIPT2013,Polunchenko+etal:SA2014} and develop a more efficient numerical method to compute the performance of the GSR procedure in the multi-cyclic setup. Specifically, the proposed method is based on the integral-equations approach and uses the standard collocation framework (see, e.g.,~\cite[Section~12.1.1]{Atkinson+Han:Book09}) in combination with a certain change-of-measure identity and a certain martingale property specific to the GSR procedure's detection statistic. As a result, the proposed method's accuracy, robustness and efficiency improve noticeably; greater efficiency is because the method can {\em simultaneously} compute both the Average Run Length (ARL) to false alarm and the Stationary Average Detection Delay (STADD). We also show that the method's rate of convergence is quadratic, and supply a tight upperbound on the method's error; the method's expected characteristics are confirmed experimentally in a specific scenario. Since the method is designed not restricted to a particular data distribution or to a specific value of the GSR detection statistic's headstart, it may help gain greater insight into the properties of the GSR procedure and aid a practitioner to set up the GSR procedure as needed while fully utilizing its potential.

The paper is a response to the call made, e.g., in~\cite{Woodall+Montgomery:JQT1999}, and then reiterated, e.g., in~\cite{Stoumbos+etal:JASA2000}, for a ``greater synthesis'' of the areas of quickest change-point detection and statistical process and quality control. While much of the paper is written using change-point detection lingo and notation, it is our hope that this work will contribute to the called for ``cross-fertilization of ideas'' from aforesaid closely interrelated fields, and thus smoothen the transition of the state-of-the-art in quickest change-point detection into the state-of-the-practice in statistical process and quality control.

The remainder of the paper is structured thus: We first formally state the problem and introduce the GSR procedure in Section~\ref{sec:GSR+properties}. The numerical method and its accuracy analysis are presented in Section~\ref{sec:performance-evaluation}. Section~\ref{sec:case-study} is devoted to a case study aimed at assessing and comparing experimentally the accuracy, robustness and convergence rate of the proposed method against those of its predecessor method offered and applied in~\cite{Moustakides+etal:CommStat09,Tartakovsky+etal:IWSM2009,Moustakides+etal:SS11}. Finally, Section~\ref{sec:conclusion} draws conclusions.

%+-----------------------------------------------------------------------------------------------+%
\section{The problem and the Generalized Shiryaev--Roberts procedure}
\label{sec:GSR+properties}

This section is intended to formally state the problem and introduce the GSR procedure.

We begin with stating the problem. Let $f(x)$ and $g(x)$ be the observations' pre- and post-change distribution densities, respectively; $g(x)\not\equiv f(x)$. Define the change-point, $0\le\nu\le\infty$, as the unknown (but not random) serial index of the final pre-change observation (so it can potentially be infinite). That is, as illustrated in Figure~\ref{fig:basic-iid-change-point}, the probability density function (pdf) of $X_n$ is $f(x)$ for $1\le n\le\nu$, and $g(x)$ for $n\ge\nu+1$.
\begin{figure}[!htb]
    \centering
    \includegraphics[width=0.8\textwidth]{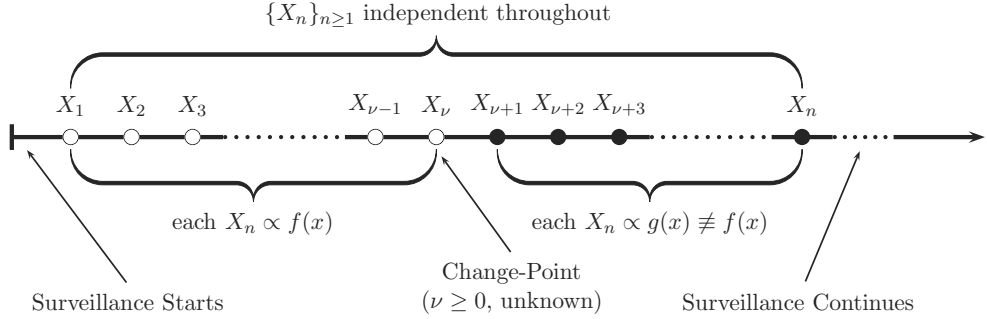}
    \caption{Basic ``minimax-ish'' setup of the quickest change-point detection problem.}
    \label{fig:basic-iid-change-point}
\end{figure}
The notation $\nu=0$ is to be understood as the case when the pdf of $X_n$ is $g(x)$ for all $n\ge1$, i.e., the data, $\{X_n\}_{n\ge1}$, are affected by change {\em ab initio}. Similarly, the notation $\nu=\infty$ is to mean that the pdf of $X_n$ is $f(x)$ for all $n\ge1$.

Let $\Pr_k$ ($\EV_k$) be the probability measure (corresponding expectation) given a known change-point $\nu=k$, where $0\le k\le\infty$. Particularly, $\Pr_\infty$ ($\EV_\infty$) is the probability measure (corresponding expectation) assuming the observations' distribution is always $f(x)$ and never changes (i.e., $\nu=\infty$). Likewise, $\Pr_0$ ($\EV_0$) is the probability measure (corresponding expectation) assuming the observations' distribution is $g(x)$ ``from the get-go'' (i.e., $\nu=0$).

From now on $\T$ will denote the stopping time associated with a generic detection procedure.

Given this ``minimax-ish'' context, the standard way to gauge the false alarm risk is through Lorden's~\cite{Lorden:AMS71} Average Run Length (ARL) to false alarm; it is defined as $\ARL(\T)\triangleq\EV_\infty[\T]$. To introduce the multi-cyclic change-point detection problem, let
\begin{align*}
\Delta(\gamma)
&\triangleq
\Bigl\{\T\colon\ARL(\T)\ge\gamma\Bigr\},\;\gamma>1,
\end{align*}
denote the class of procedures, $\T$, with the ARL to false alarm at least $\gamma>1$, a pre-selected tolerance level. Suppose now that it is of utmost importance to detect the change as quickly as possible, even at the expense of raising many false alarms (using a repeated application of the same procedure) before the change occurs. Put otherwise, in exchange for the assurance that the change will be detected with maximal speed, one agrees to go through a ``storm'' of false alarms along the way (the false alarms are ensued from repeatedly applying the same procedure, starting from scratch after each false alarm). This scenario is shown in Figure~\ref{fig:multi-cyclic-idea}.
\begin{figure*}[!t]
    \centering
    \subfloat[An example of the behavior of a process of interest with a change in mean at time $\nu$.]{
        \includegraphics[width=0.9\textwidth]{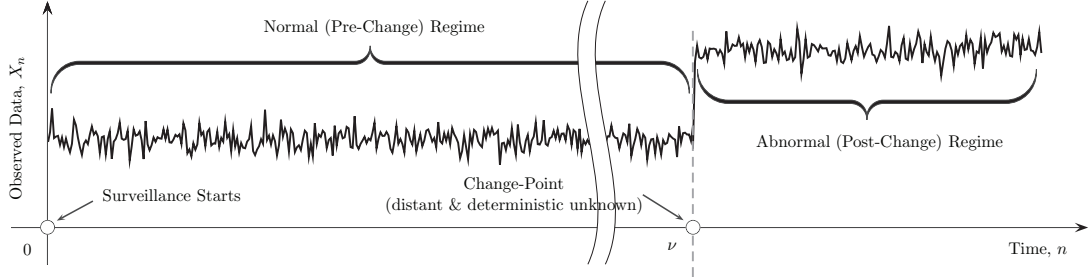}
    }\\ % /subfigure
    \subfloat[Typical behavior of the detection statistic in the multi-cyclic mode.]{
        \includegraphics[width=0.9\textwidth]{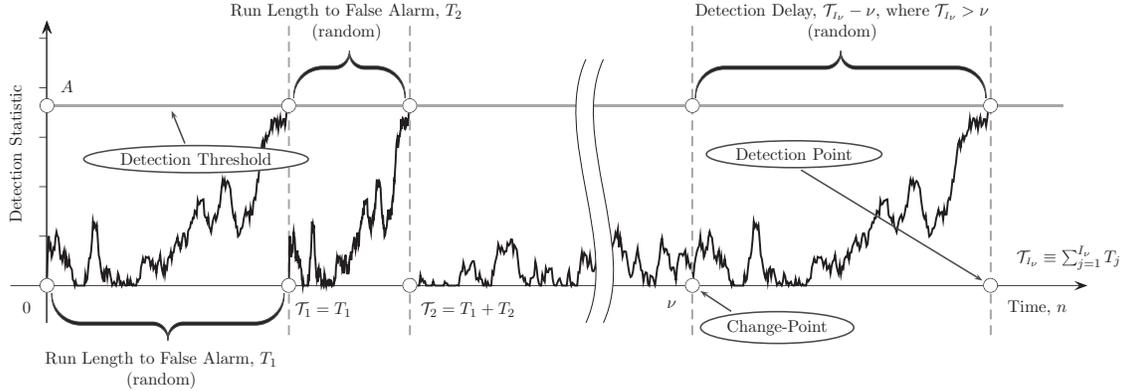}
    } % /subfigure
    % figure caption to be below the figure
    \caption{Multi-cyclic change-point detection in a stationary regime.}
    \label{fig:multi-cyclic-idea}
\end{figure*}
Formally, let $T_1,T_2,\ldots$ be sequential independent repetitions of the same stopping time, $\T$, and let ${\cal T}_j\triangleq T_1+T_2+\cdots+T_j$, $j\ge1$, be the time of the $j$-th alarm. Define $I_\nu\triangleq\min\{j\ge1\colon {\cal T}_j>\nu\}$ so that ${\cal T}_{\scriptscriptstyle I_\nu}$ is the time of detection of a true change that occurs at time moment $\nu$ after $I_\nu-1$ false alarms had been raised. One can then view the difference ${\cal T}_{\scriptscriptstyle I_\nu}-\nu(\ge0)$ as the detection delay. Let
\begin{align}\label{eq:STADD-def}
\STADD(\T)
&\triangleq
\lim_{\nu\to\infty}\EV_\nu[{\cal T}_{\scriptscriptstyle I_\nu}-\nu]
\end{align}
be the limiting value of the Average Detection Delay (ADD) referred to as the {\em Stationary ADD} (STADD). We hasten to note that the STADD and the Steady-State ADD (SSADD), or the Steady-State ARL, a detection delay measure popular in the areas of statistical process and quality control, are {\em not} the same thing; we will comment more on the difference in the end of this section. The multi-cyclic change-point detection problem is:
\begin{align}\label{eq:multi-cyclic-problem}
\text{to find $\T_{\mathrm{opt}}\in\Delta(\gamma)$ such that $\STADD(\T_{\mathrm{opt}})=\inf_{\T\in\Delta(\gamma)}\STADD(\T)$ for every $\gamma>1$.}
\end{align}

As can be seen from the description, the multi-cyclic formulation is instrumental in detecting a change that takes place in a distant future (i.e., $\nu$ is large), and is preceded by a stationary flow of false detections, each with a cost much smaller than that of missing the change by a single observation. By way of example, such scenarios are encountered, e.g., in cybersecurity~\cite{Polunchenko+etal:SA2012,Tartakovsy+etal:IEEE-JSTSP2013} and in the economic design of control charts~\cite{Duncan:JASA1956,Montgomery:JQT1980,Lorenzen+Vance:T1986,Ho+Case:JQT1994}.

Since the STADD is defined as a limit, the natural question is how does one evaluate it in practice? The answer is provided by the fact that the multi-cyclic formulation~\eqref{eq:multi-cyclic-problem} and the generalized Bayesian formulation of the change-point detection problem are completely equivalent to one another; see, e.g.,~\cite{Shiryaev:SMD61,Shiryaev:TPA63,Pollak+Tartakovsky:SS09}. A recent overview of all major formulations of the change-point detection problem can be found, e.g., in~\cite{Tartakovsky+Moustakides:SA10,Polunchenko+Tartakovsky:MCAP2012,Polunchenko+etal:JSM2013}. Specifically, the generalized Bayesian formulation is a limiting case of the Bayesian formulation with an (improper) uniform prior distribution imposed on the change-point, $\nu$. Under this assumption, the objective of the generalized Bayesian formulation is to find a procedure, $\T_{\mathrm{opt}}\in\Delta(\gamma)$, that minimizes the so-called Relative Integral ADD (RIADD) inside class $\Delta(\gamma)$ for every $\gamma>1$. Formally, the RIADD is defined as
\begin{align}\label{eq:RIADD-def}
\RIADD(\T)
&\triangleq
\IADD(\T)/\ARL(\T),
\end{align}
where
\begin{align}\label{eq:IADD-def}
\IADD(\T)
&\triangleq
\sum_{k=0}^\infty\EV_k[\max\{0,\T-k\}]
\end{align}
is the so-called Integral ADD (IADD)\footnote{The objective of the generalized Bayesian formulation is also often stated as ``to find $\T_{\mathrm{opt}}\in\Delta(\gamma)$ that minimizes the IADD inside class $\Delta(\gamma)$ for every $\gamma>1$. Due to the structure of the class $\Delta(\gamma)$ it is the same as attempting to minimize the RIADD inside that class.}. The equivalence of the multi-cyclic formulation and the generalized Bayesian formulation is in the statement that $\STADD(\T)\equiv\RIADD(\T)$ for any detection procedure, $\T$. For a proof see, e.g.,~\cite[Theorem~2]{Pollak+Tartakovsky:SS09} or~\cite{Shiryaev+Zryumov:Khabanov2010}; in continuous time the same result was obtained by Shiryaev, e.g., in~\cite{Shiryaev:SMD61,Shiryaev:TPA63,Shiryaev:Bachelier2002,Feinberg+Shiryaev:SD2006}. Hence, the STADD does not have to be computed as the limit~\eqref{eq:STADD-def}; instead, it can be evaluated as the RIADD through~\eqref{eq:RIADD-def}-\eqref{eq:IADD-def}. The specifics are discussed in Section~\ref{sec:performance-evaluation}. See also, e.g.,~\cite{Moustakides+etal:CommStat09,Tartakovsky+etal:IWSM2009,Moustakides+etal:SS11}.

It is shown in~\cite{Pollak+Tartakovsky:SS09,Shiryaev+Zryumov:Khabanov2010} that the multi-cyclic change-point detection problem~\eqref{eq:multi-cyclic-problem} is solved by the (original) Shiryaev--Roberts (SR) procedure~\cite{Shiryaev:SMD61,Shiryaev:TPA63,Roberts:T66}; incidentally, the comparative performance analysis offered in~\cite{Moustakides+etal:CommStat09,Tartakovsky+etal:IWSM2009} demonstrates that both the CUSUM scheme and the EWMA chart are outperformed (in the multi-cyclic sense) by the SR procedure. We now introduce the SR procedure. To that end, since the SR procedure is likelihood ratio-based, we first construct the corresponding likelihood ratio (LR).

Let $\mathcal{H}_k\colon\nu=k$ for $0\le k<\infty$ be the hypothesis that the change takes place at time moment $\nu=k$ for $0\le k<\infty$. Let $\mathcal{H}_{\infty}\colon\nu=\infty$ be the hypothesis that no change ever occurs (i.e., $\nu=\infty$). The joint distribution densities of the sample $\boldsymbol{X}_{1:n}\triangleq(X_1,\ldots,X_n)$, $n\ge1$, under each of these hypotheses are given by
\begin{align*}
p(\boldsymbol{X}_{1:n}|\mathcal{H}_{\infty})
&=
\prod_{j=1}^n f(X_j)\;\;\text{and}\;\;
p(\boldsymbol{X}_{1:n}|\mathcal{H}_k)
=
\prod_{j=1}^k f(X_j)\prod_{j=k+1}^n
g(X_j),\;\text{for $k<n$},
\end{align*}
with $p(\boldsymbol{X}_{1:n}|\mathcal{H}_{\infty})=p(\boldsymbol{X}_{1:n}|\mathcal{H}_k)$ for $k\ge n$. The corresponding LR therefore is
\begin{align*}
\LR_{1:n,\nu=k}
&\triangleq
\frac{p(\boldsymbol{X}_{1:n}|\mathcal{H}_k)}{p(\boldsymbol{X}_{1:n}|\mathcal{H}_{\infty})}
=\prod_{j=k+1}^n\LR_j,\;\text{for $k<n$},
\end{align*}
where from now on $\LR_n\triangleq g(X_n)/f(X_n)$ is the ``instantaneous'' LR for the $n$-th observation, $X_n$.

We now make an observation that will play an important role in the sequel. Let $P_d^{\LR}(t)\triangleq\Pr_d(\LR_1\le t)$, $t\ge0$, $d=\{0,\infty\}$, denote the cdf of the LR under measure $\Pr_d$, $d=\{0,\infty\}$, respectively. As the LR is the Radon--Nikod{\'y}m derivative of measure $\Pr_0$ with respect to measure $\Pr_\infty$, one can conclude that
\begin{align}\label{eq:change-of-measure-identity}
dP_0^{\LR}(t)
&=
t\,dP_\infty^{\LR}(t),\;\;t\ge0;
\end{align}
cf.~\cite{Polunchenko+etal:MIPT2013,Polunchenko+etal:SA2014}. It is assumed that measures $\Pr_0$ and $\Pr_\infty$ are mutually absolutely continuous. We will use this change-of-measure identity heavily in Section~\ref{sec:performance-evaluation} to improve the accuracy, rate of convergence, and efficiency of our numerical method.

Formally, the original SR procedure~\cite{Shiryaev:SMD61,Shiryaev:TPA63,Roberts:T66} is defined as the stopping time
\begin{align}\label{eq:T-SR-def}
\mathcal{S}_A
&\triangleq
\inf\big\{n\ge1\colon R_n\ge A\big\},
\end{align}
where $A>0$ is a detection threshold used to control the false alarm risk, and
\begin{align}\label{eq:Rn-SR-def}
R_n
&\triangleq
\sum_{k=1}^n\LR_{1:n,\nu=k}=\sum_{k=1}^n \prod_{i=k}^n\LR_i,\; n\ge1,
\end{align}
is the SR detection statistic; here and throughout the rest of the paper in every definition of a detection procedure we will assume that $\inf\{\varnothing\}=\infty$. Note the recursion
\begin{align}\label{SRstatrec}
R_{n+1}
&=
(1+R_{n})\LR_{n+1}\;\text{for}\; n=0,1,\ldots\;\text{with}\; R_0=0,
\end{align}
and we stress that $R_0=0$, i.e., the SR statistic starts from zero. Observe now that $\{R_n-n\}_{n\ge0}$ is a zero-mean $\Pr_\infty$-martingale, i.e., $\EV_\infty[R_n-n]=0$ for any $n\ge0$. From this and the Optional stopping theorem (see, e.g.,~\cite[Subsection~2.3.2]{Poor+Hadjiliadis:Book08} or~\cite[Chapter~VII]{Shiryaev:Book1995}), one can conclude that $\EV_\infty[R_{\mathcal{S}_A}-\mathcal{S}_A]=0$, whence $\ARL(\mathcal{S}_A)\triangleq\EV_\infty[\mathcal{S}_A]=\EV_\infty[R_{\mathcal{S}_A}]\ge A$. It is now easy for one to set the detection threshold, $A$, so as to ensure $\ARL(\mathcal{S}_A)\ge\gamma$ for any desired $\gamma>1$. More specifically, it can be shown~\cite{Pollak:AS87} that $\ARL(\mathcal{S}_A)=(A/\xi)[1+o(1)]$, as $\gamma\to\infty$, where $\xi\in(0,1)$ is the limiting exponential overshoot, a model-dependent constant that can be computed using nonlinear renewal theory~\cite{Siegmund:Book85,Woodroofe:Book82}. For practical purposes, the approximation $\ARL(\mathcal{S}_A)\approx A/\xi$ is known to be extremely accurate under broad conditions. More importantly, as shown in~\cite{Pollak+Tartakovsky:SS09,Shiryaev+Zryumov:Khabanov2010}, the SR procedure is {\em exactly} $\STADD(\T)$-optimal, i.e., formally: $\STADD(\mathcal{S}_{A_\gamma})=\inf_{\T\in\Delta(\gamma)}\STADD(\T)$ for every $\gamma>1$, where $A_\gamma>0$ is the solution of the equation $\ARL(\mathcal{S}_{A_\gamma})=\gamma$.

This strong optimality property of the SR procedure~\eqref{eq:T-SR-def}-\eqref{SRstatrec} was recently generalized in~\cite[Lemma~1]{Polunchenko+Tartakovsky:AS10} where the SR procedure was allowed to have a headstart. This version of the SR procedure is known as the Shiryaev--Roberts--$r$ (SR--$r$) procedure, and it was proposed in~\cite{Moustakides+etal:SS11}. Specifically, the SR--$r$ procedure regards starting off the original SR procedure~\eqref{eq:T-SR-def}-\eqref{SRstatrec} at a fixed (but specially designed) $R_0^r=r$, $r\ge0$, i.e., $r\ge0$ is a headstart. This is similar to the idea proposed earlier in~\cite{Lucas+Crosier:T1982} for the CUSUM scheme. However, it turns out that, unlike for the CUSUM scheme, giving the SR procedure a headstart is practically ``putting it on steroids'': the gain in performance far exceeds that observed in~\cite{Lucas+Crosier:T1982} for the CUSUM scheme.

Formally, the SR--$r$ procedure is defined as the stopping time
\begin{align}\label{SR-rst}
\mathcal{S}_{A}^r
&=
\inf\{n\ge1\colon R_n^r \ge A\},
\end{align}
where again $A>0$ and is used to control the ARL to false alarm, and
\begin{align}\label{SR-rstat}
R_{n+1}^r
&=
(1+R_{n}^r)\LR_{n+1}\;\text{for}\; n=0,1,\ldots\;\text{with}\; R_0^r=r\ge0,
\end{align}
and we remark that for $r=0$ the SR--$r$ procedure becomes the original SR procedure~\eqref{eq:T-SR-def}-\eqref{SRstatrec}. For this reason from now on we will collectively refer to both procedures as the Generalized SR (GSR) procedure, following the terminology used in~\cite{Tartakovsky+etal:TPA2012}. Observe that $\{R_n^r-n-r\}_{n\ge0}$ is a zero-mean $\Pr_\infty$-martingale, i.e., $\EV_\infty[R_n^r-n-r]=0$ for all $n\ge0$ and all $r$. As a result, one can generalize~\cite{Pollak:AS87} to conclude that
\begin{align}\label{eq:SRr-ARL-approx}
\ARL(\mathcal{S}_A^r)
&\approx
\dfrac{A}{\xi}-r\;\text{for sufficiently large $A>0$};
\end{align}
here $\xi\in(0,1)$ is again the limiting exponential overshoot. This approximation is also quite accurate under broad conditions. More importantly, it is shown in~\cite[Lemma~1]{Polunchenko+Tartakovsky:AS10} that the SR--$r$ procedure minimizes the generalized STADD
\begin{align}\label{eq:STADD-gen-def}
\STADD(\T)
&\triangleq
\left(r\EV_0[\T]+\IADD(\T)\right)\left/\left(\ARL(\T)+r\right)\right.
\end{align}
within class $\Delta(\gamma)$; here $\IADD(\T)$ is as in~\eqref{eq:IADD-def} above. From now on we will consider only the generalized STADD. It is direct to see that for $r=0$ the generalized STADD coincides with the RIADD given by~\eqref{eq:RIADD-def}. Formally, from~\cite[Lemma~1]{Polunchenko+Tartakovsky:AS10} we have that $\STADD(\mathcal{S}_A^r)=\inf_{\T\in\Delta(\gamma)}\STADD(\T)$ for every $\gamma>1$, where $A$ and $r$ are such that $\ARL(\mathcal{S}_A^r)=\gamma$ is true; for $r=0$ this reduces to the result established in~\cite{Pollak+Tartakovsky:SS09,Shiryaev+Zryumov:Khabanov2010} for the original SR procedure.

We conclude this section with a remark on the difference between the STADD given by~\eqref{eq:STADD-gen-def} and the Steady-State ADD (SSADD); the latter is often called the Steady-State ARL, and is a control chart performance metric popular in the area of quality control as metric less prone to the adverse ``inertia effect''~\cite{Yashchin:IBM-JRDMC1987,Yashchin:ISR1993,Woodall+Adams:HSMES1998}. Formally, the SSADD is defined as $\SSADD(\T)\triangleq\lim_{k\to\infty}\EV_k[\T-k|\T>k]$; see, e.g.,~\cite{Knoth:FSQC2006}. The principal difference between the STADD and the SSADD is that the SSADD is assuming the procedure of choice, $\T$, is applied only {\em once}, whereas $\STADD(\T)$ is assuming {\em repetitive} and independent application of $\T$. Hence, the steady-state regime involved in $\SSADD(\T)$ is different from the stationary regime involved in $\STADD(\T)$: the former is pertaining to the detection statistic, while the latter is pertaining to the change-point.

%+-----------------------------------------------------------------------------------------------+%
\section{Performance evaluation}
\label{sec:performance-evaluation}

We now develop a numerical method to evaluate the performance of the GSR procedure~\eqref{SR-rst}-\eqref{SR-rstat} in the multi-cyclic setup~\eqref{eq:multi-cyclic-problem}. Specifically, we ``gear'' the method toward numerical computation of two antagonistic performance measures associated with the GSR stopping time $\mathcal{S}_A^r$:\begin{inparaenum}[\itshape a)]\item the usual ``in-control'' ARL to false alarm, i.e., $\ARL(\mathcal{S}_A^r)\triangleq\EV_\infty[\mathcal{S}_A^r]$, and \item $\STADD(\mathcal{S}_A^r)$\end{inparaenum}, i.e., the Stationary Average Detection Delay~\eqref{eq:STADD-gen-def}. More concretely, we first derive an integral equation for each performance measure involved. Using the change-of-measure identity~\eqref{eq:change-of-measure-identity} we then show that the equation for $\ARL(\mathcal{S}_A^r)$ and that for the numerator of $\STADD(\mathcal{S}_A^r)$---see~\eqref{eq:STADD-gen-def}---differ only by the right-hand side (which is completely known for either equation). As a result, both equations can be solved {\em concurrently}. Finally, we present our numerical method to ({\em simultaneously}) solve the obtained equations, and offer an analysis of the method's accuracy and rate of convergence.

The proposed method is a build-up over one previously proposed in~\cite{Tartakovsky+etal:IWSM2009,Moustakides+etal:CommStat09,Moustakides+etal:SS11} and recently extended in~\cite{Polunchenko+etal:SA2014}; see also~\cite{Polunchenko+etal:MIPT2013}.

%+-----------------------------------------------------------------------------------------------+%
\subsection{Integral equations}
\label{ssec:integral-equations}

We begin with notation and assumptions. First recall $\LR_n\triangleq g(X_n)/f(X_n)$, i.e., the ``instantaneous'' LR for the $n$-th data point, $X_n$. For simplicity, $\LR_1$ will be assumed absolutely continuous, although at an additional effort the case of strictly non-arithmetic $\LR_1$ can be handled as well. Let ${P}_d^{\LR}(t)\triangleq\Pr_d(\LR_1\le t)$, $d=\{0,\infty\}$, $t\ge0$, be the cdf of the LR under the measure $\Pr_d$, $d=\{0,\infty\}$. Also, denote
\begin{align}\label{eq:K-def}
{K}_d(x,y)
&\triangleq
\frac{\partial}{\partial y}\Pr_d(R_{n+1}^{r}\le y|R_n^{r}=x)
=
\frac{\partial}{\partial y}{P}_d^{\LR}\left(\frac{y}{1+x}\right),
\; d=\{0,\infty\},
\end{align}
the {\em transition probability density kernel} for the (stationary) Markov process $\{R_n^r\}_{n\ge0}$.

We now note that from the change-of-measure identity $d{P}_0^{\LR}(t)=t\,d{P}_{\infty}^{\LR}(t)$, $t\ge0$, mentioned earlier, and definition~\eqref{eq:K-def} one can readily deduce that $(1+x)\,{K}_0(x,y)=y\,{K}_{\infty}(x,y)$; cf.~\cite{Polunchenko+etal:MIPT2013,Polunchenko+etal:SA2014}. This can be used, e.g., as a ``shortcut'' in deriving the formula for ${K}_0(x,y)$ from that for ${K}_{\infty}(x,y)$, or the other way around---whichever one of the two is found first. More importantly, as will be shown in Theorem~\ref{thm:lwr-bnd-int-eqn} below, using $(1+x)\,{K}_0(x,y)=y\,{K}_{\infty}(x,y)$, one can ``tie'' $\ARL(\mathcal{S}_A^{r=x})$ and $\STADD(\mathcal{S}_A^{r=x})$ to one another in such a way so that both can be computed {\em simultaneously}, with ${K}_0(x,y)$ completely eliminated. This result will then be used to design our numerical method in the next subsection. Last but not least, as done in~\cite{Polunchenko+etal:MIPT2013,Polunchenko+etal:SA2014}, we will also use this connection between ${K}_0(x,y)$ and ${K}_{\infty}(x,y)$ to improve the method's accuracy and rate of convergence; see Subsection~\ref{ssec:numerical-method+accuracy} below.

We now state the first equation of interest. Let $R_0^{r=x}=x\ge0$ be fixed. For notational brevity, from now on let $\ell(x,A)\triangleq\ARL(\mathcal{S}_A^{r=x})\triangleq\EV_{\infty}[\mathcal{S}_A^{r=x}]$; we reiterate that this expectation is conditional on $R_0^{r=x}=x$. Using the fact that $\{R_n^{r=x}\}_{n\ge0}$ is Markovian, it can be shown that $\ell(x,A)$ is governed by the equation
\begin{align}\label{eq:ARL-int-eqn}
\ell(x,A)
&=
1+\int_0^A{K}_{\infty}(x,y)\,\ell(y,A)\,dy;
\end{align}
cf.~\cite{Moustakides+etal:SS11}. %We also recall the approximation $\ell(x,A)\approx A/v-x$ valid for $A$ sufficiently large. We also note that $\ell(x,0)=1$ for any $x\ge 0$, and $\ell(x,A)\to1$, as $x\to\infty$, for any $A>0$.

Next, introduce $\delta_k(x,A)\triangleq\EV_k[(\mathcal{S}_A^{r=x}-k)^+]$, $k\ge0$. For $k=0$ observe that
\begin{align}\label{eq:ADD0-int-eqn}
\delta_0(x,A)
&=
1+\int_0^A{K}_{0}(x,y)\,\delta_0(y,A)\,dy,
\end{align}
which is an exact ``copy'' of equation~\eqref{eq:ARL-int-eqn} except that ${K}_\infty(x,y)$ is replaced with ${K}_0(x,y)$; cf.~\cite{Moustakides+etal:SS11}.

For $k\ge1$, since $\{R_n^{r=x}\}_{n\ge0}$ is Markovian, one can establish the recursion
\begin{align}\label{eq:ADDk-recursion}
\delta_{k+1}(x,A)
&=
\int_0^A{K}_{\infty}(x,y)\,\delta_k(y,A)\,dy,\; k\ge0,
\end{align}
with $\delta_0(x,A)$ first found from equation~\eqref{eq:ADD0-int-eqn}; cf.~\cite{Moustakides+etal:SS11}. Using this recursion one can generate the entire functional sequence $\{\delta_k(x,A)\}_{k\ge0}$ by repetitive application of the linear integral operator
\begin{align*}
\mathcal{K}_{\infty}\circ u
&\triangleq
[\mathcal{K}_{\infty}\circ u](x)
\triangleq
\int_0^A {K}_{\infty}(x,y)\,u(y)\,dy,
\end{align*}
where $u(x)$ is assumed to be sufficiently smooth inside the interval $[0,A]$. Temporarily deferring formal discussion of this operator's properties, note that using this operator notation, recursion~\eqref{eq:ADDk-recursion} can be rewritten as $\delta_{k+1}=\mathcal{K}_{\infty}\circ \delta_{k}$, $k\ge0$, or equivalently, as $\delta_{k}=\mathcal{K}_{\infty}^{k}\circ \delta_{0}$, $k\ge0$, where
\begin{align*}
\mathcal{K}_{\infty}^{k}\circ u
&\triangleq
\underbrace{\mathcal{K}_{\infty}\circ\cdots\circ\mathcal{K}_{\infty}}_{\text{$k$ times}}\circ\,u\;\text{for}\;k\ge1,
\end{align*}
and $\mathcal{K}_{\infty}^{0}$ is the identity operator from now on denoted as $\mathbb{I}$, i.e., $\mathcal{K}_{\infty}^{0}\circ u=\mathbb{I}\circ u\triangleq u$. Similarly, in the operator form, equation~\eqref{eq:ARL-int-eqn} can be rewritten as $\ell=1+\mathcal{K}_\infty\circ\ell$, and equation~\eqref{eq:ADD0-int-eqn} can be rewritten as $\delta_0=1+\mathcal{K}_0\circ\delta_0$.

We now note that the sequence $\{\delta_k(x,A)\}_{k\ge0}$ can be used to derive the equation for $\IADD(\mathcal{S}_A^{r=x})$ defined by~\eqref{eq:IADD-def}. Let $\psi(x,A)\triangleq\IADD(\mathcal{S}_A^{r=x})$, and observe that
\begin{align}\label{eq:IADD-Neumann-series}
\psi
&\triangleq
\sum_{k\ge0}\delta_k=\sum_{k\ge0}\mathcal{K}_\infty^k\circ\delta_0
=
\left(\,\sum_{k\ge0}\mathcal{K}_\infty^k\right)\circ\delta_0
=
(\,\mathbb{I}-\mathcal{K}_\infty)^{-1}\circ\delta_0,
\end{align}
whence
\begin{align}\label{eq:IADD-int-eqn}
\psi(x,A)
&=
\delta_0(x,A)+\int_0^A{K}_\infty(x,y)\,\psi(y,A)\,dy;
\end{align}
cf.~\cite{Moustakides+etal:SS11}. The implicit use of the geometric series convergence theorem in~\eqref{eq:IADD-Neumann-series} is justified by the fact that the spectral radius of the operator $\mathcal{K}_\infty$ is strictly less than 1; see, e.g.,~\cite{Moustakides+etal:SS11}.

At this point, with equations for $\ell(x,A)$, $\delta_0(x,A)$, and for $\psi(x,A)$ obtained, one can compute $\STADD(\mathcal{S}_A^{r=x})$ through
\begin{align}\label{eq:STADD-SADD-eqn}
\STADD(\mathcal{S}_A^{r=x})
&=
[x\,\delta_0(x,A)+\psi(x,A)]/[\ell(x,A)+x],
\end{align}
for any $x\ge0$; in particular, for $x=0$ this gives $\STADD(\mathcal{S}_A)$, i.e., the STADD for the original SR procedure~\eqref{eq:T-SR-def}-\eqref{eq:Rn-SR-def}. Thus, it may seem that the strategy to compute $\STADD(\mathcal{S}_A^{r=x})$ is to first compute $\delta_0(x,A)$ by solving equation~\eqref{eq:ADD0-int-eqn}, then use the obtained $\delta_0(x,A)$ to compute $\psi(x,A)$ by solving equation~\eqref{eq:IADD-int-eqn}, independently solve equation~\eqref{eq:ARL-int-eqn} to get $\ell(x,A)$, and finally plug all these into~\eqref{eq:STADD-SADD-eqn} to get $\STADD(\mathcal{S}_A^{r=x})$. Precisely this strategy was employed in~\cite{Tartakovsky+etal:IWSM2009,Moustakides+etal:CommStat09,Moustakides+etal:SS11}. However, we will now show that the computation of $\STADD(\mathcal{S}_A^{r=x})$ can be made much simpler. Let $\Xi(x,A)\triangleq x\,\delta_0(x,A)+\psi(x,A)$ so that $\STADD(\mathcal{S}_A^{r=x})=\Xi(x,A)/[\ell(x,A)+x]$.

\begin{theorem}\label{thm:lwr-bnd-int-eqn}
\begin{align}\label{eq:psi+delta-int-eqn}
\Xi(x,A)
&=
1+x+\int_0^A K_\infty(x,y)\,\Xi(y,A)\,dy.
\end{align}
\end{theorem}
\begin{proof}
First, consider equation~\eqref{eq:ADD0-int-eqn} and multiply it through by $(1+x)$ to obtain
\begin{align*}
(1+x)\,\delta_0(x,A)
&=
1+x+\int_0^A(1+x)\,{K}_0(x,y)\,\delta_0(y,A)\,dy,
\end{align*}
which using the change-of-measure identity $(1+x)\,{K}_0(x,y)=y\,{K}_\infty(x,y)$ is equivalent to
\begin{align}\label{eq:delta-x-int-eqn}
(1+x)\,\delta_0(x,A)
&=
1+x+\int_0^A\,{K}_\infty(x,y)\,y\,\delta_0(y,A)\,dy.
\end{align}

Next, by adding
\begin{align*}
\int_0^A{K}_\infty(x,y)\,\psi(y,A)\,dy
\end{align*}
to both sides of~\eqref{eq:delta-x-int-eqn}, we obtain
\begin{align*}
(1+x)\,\delta_0(x,A)&+\int_0^A{K}_\infty(x,y)\,\psi(y,A)\,dy
=\\
&1+x+\int_0^A{K}_\infty(x,y)\,\psi(y,A)\,dy+\int_0^A\,{K}_\infty(x,y)\,y\,\delta_0(y,A)\,dy,
\end{align*}
which after some algebra becomes
\begin{align*}
\left(\delta_0(x,A)+\int_0^A{K}_\infty(x,y)\,\psi(y,A)\,dy\right)
&+x\,\delta_0(x,A)
=\\
&1+x+\int_0^A{K}_\infty(x,y)\,[\psi(y,A)+y\,\delta_0(y,A)]\,dy.
\end{align*}

Finally, note that the expression in parentheses in the left-hand side above is the right-hand side of equation~\eqref{eq:IADD-int-eqn}, i.e., it is equal to $\psi(x,A)$. Hence, recalling that $\Xi(x,A)\triangleq \psi(x,A)+x\,\delta_0(x,A)$, we arrive at the desired equation for $\Xi(x,A)$, i.e., at equation~\eqref{eq:psi+delta-int-eqn}.
\end{proof}

Using Theorem~\ref{thm:lwr-bnd-int-eqn}, i.e., equation~\eqref{eq:psi+delta-int-eqn}, one can compute $\ell(x,A)$ and $\Xi(x,A)$ {\em simultaneously} and without having to compute $\delta_0(x,A)$ and $\psi(x,A)$ at all. Specifically, since~\eqref{eq:ARL-int-eqn} and~\eqref{eq:psi+delta-int-eqn} can be rewritten, respectively, as $(\mathbb{I}-\mathcal{K}_\infty)\circ\ell=1$ and $(\mathbb{I}-\mathcal{K}_\infty)\circ\Xi=1+x$ in the operator form, one can see that both have the same integral operator in the left-hand side, and the right-hand side of either is completely known and does not require any preliminary evaluation. Thus, to evaluate the $\ARL(\mathcal{S}_A^{r=x})$ and $\STADD(\mathcal{S}_A^{r=x})$ one is effectively to solve two equations $(\mathbb{I}-\mathcal{K}_\infty)\circ u = 1$ and $(\mathbb{I}-\mathcal{K}_\infty)\circ u = x$, which can be done {\em simultaneously}. This is an improvement over the method proposed and used earlier in~\cite{Tartakovsky+etal:IWSM2009,Moustakides+etal:CommStat09,Moustakides+etal:SS11}. It is also an extension of the method proposed recently in~\cite{Polunchenko+etal:SA2014}; see also, e.g.,~\cite{Polunchenko+etal:MIPT2013}.

Combined, equations~\eqref{eq:ARL-int-eqn} and~\eqref{eq:psi+delta-int-eqn} form a ``complete package'' to compute any of the desired performance characteristics of the GSR procedure. The question to be considered next is that of computing these characteristics in practice.

%+-----------------------------------------------------------------------------------------------+%
\subsection{The numerical method and its accuracy analysis}
\label{ssec:numerical-method+accuracy}

We now turn to the question of solving the main equations---\eqref{eq:ARL-int-eqn} and~\eqref{eq:psi+delta-int-eqn}---presented in the preceding subsection. To this end, following~\cite{Tartakovsky+etal:IWSM2009,Moustakides+etal:CommStat09,Moustakides+etal:SS11,Polunchenko+etal:MIPT2013,Polunchenko+etal:SA2014}, observe first that both equations are renewal-type equations of the form
\begin{align}\label{eq:general-fredholm-eqn}
u(x)
&=
\upsilon(x)+\int_0^A{K}_\infty(x,y)\,u(y)\,dy,
\end{align}
where $\upsilon(x)$ is a given (known) function, $K_\infty(x,y)$ is as in~\eqref{eq:K-def}, and $u(x)$ is the unknown; note that while $u(x)$ does depend on the upper limit of integration, $A>0$, for notational simplicity, we will no longer emphasize that, and use the notation $u(x)$ instead of $u(x,A)$.

To see that equation~\eqref{eq:general-fredholm-eqn} is an ``umbrella'' equation for equations~\eqref{eq:ARL-int-eqn} and~\eqref{eq:psi+delta-int-eqn}, observe that, e.g., to obtain equation~\eqref{eq:ARL-int-eqn} on the ARL to false alarm, it suffices to set $\upsilon(x)\equiv 1$ for any $x\in\mathbb{R}$. Similarly, choosing $\upsilon(x)=1+x$ will yield equation~\eqref{eq:psi+delta-int-eqn}. Thus, any method to solve~\eqref{eq:general-fredholm-eqn} for a given $\upsilon(x)$ can be applied to solve~\eqref{eq:ARL-int-eqn} and~\eqref{eq:psi+delta-int-eqn} as well. The problem however, is that~\eqref{eq:general-fredholm-eqn} is a Fredholm integral equation of the second kind, and such equations seldom allow for an analytical solution. Hence, a numerical approach is in order and the aim of this subsection is to present one.

We first set the underlying space for the problem. Let $\mathcal{X}=\mathbb{C}[0,A]$ be the space of continuous functions over the interval $[0,A]$. Equip $\mathcal{X}$ with the usual uniform norm $\|u\|_\infty\triangleq\max_{x\in[0,A]}|u(x)|$. We will assume that ${P}_\infty^{\LR}(t)$ and the unknown function $u(x)$ are both continuous and well-behaved, i.e., both are differentiable as far as necessary. Under these assumptions $\mathcal{K}_\infty$ is a bounded linear operator from $\mathcal{X}$ into $\mathcal{X}$, equipped with the usual $\mathrm{L}_\infty$-norm:
\begin{align*}
\|\mathcal{K}_\infty\|_\infty
&\triangleq
\sup_{x\in[0,A]}\int_0^A\abs{{K}_\infty(x,y)}dy.
\end{align*}

It can be shown~\cite{Moustakides+etal:SS11} that $\|\mathcal{K}_\infty\|_\infty<1$. Thus, one can apply the Fredholm alternative~\cite[Theorem~2.8.10]{Atkinson+Han:Book09} to deduce that $(\mathbb{I}-\mathcal{K}_\infty)^{-1}$ is a bounded operator, and subsequently conclude that~\eqref{eq:general-fredholm-eqn} does have a solution and it is unique for any given $\upsilon(x)$.

To solve~\eqref{eq:general-fredholm-eqn} we propose to use the collocation method~\cite[Section~12.1.1]{Atkinson+Han:Book09}. The idea of this method is to first approximate the sought function, $u(x)$, as
\begin{align}\label{eq:u_N-def}
u_N(x)
&=
\sum_{j=1}^N u_{j,N}\,\phi_j(x),\;\; N\ge1,
\end{align}
where $\{u_{j,N}\}_{1\le j\le N}$ are constant coefficients to be determined, and $\{\phi_j(x)\}_{1\le j\le N}$ are suitably chosen (known) basis functions. For any such basis and any given $\{u_{j,N}\}_{1\le j\le N}$, substitution of $u_N(x)$ into the equation will yield a residual $r_N\triangleq u_N-\mathcal{K}_{\infty}\circ u_N-\upsilon$. Unless the true solution $u(x)$ itself is a linear combination of the basis functions $\{\phi_j(x)\}_{1\le j\le N}$, no choice of the coefficients $\{u_{j,N}\}_{1\le j\le N}$ will make the residual identically zero uniformly at all $x\in[0,A]$. However, by requiring $r_N(x)$ to be zero at some $\{z_j\}_{1\le j\le N}$, where $z_j\in[0,A]$ for all $j=1,2,\ldots,N$, one can achieve a certain level of proximity of the residual to zero. These points, $\{z_j\}_{1\le j\le N}$, are called the collocation nodes, and their choice is discussed below. As a result, we obtain the following system of $N$ algebraic equations on the coefficients $u_{j,N}$
\begin{align}\label{eq:fredholm-lin-system}
\boldsymbol{u}_N
&=
\boldsymbol{\upsilon}
+
\mathcal{\boldsymbol{K}}_\infty \boldsymbol{u}_N,
\end{align}
where $\boldsymbol{u}_N\triangleq[u_{1,N},\ldots,u_{N,N}]^\top$, $\boldsymbol{\upsilon}\triangleq[\upsilon(z_1),\ldots,\upsilon(z_N)]^\top$, and  $\boldsymbol{K}_\infty$ is a matrix of size $N$-by-$N$ whose $(i,j)$-th element is as follows:
\begin{align}\label{eq:def-matrix-K}
(\boldsymbol{K}_\infty)_{i,j}
&\triangleq
\int_0^A{K}_\infty(z_i,y)\,\phi_j(y)\,dy,\;\;1\le i,j\le N.
\end{align}

For the system of linear equations~\eqref{eq:fredholm-lin-system} to have one and only one solution, the functions $\{\phi_j(x)\}_{1\le j\le N}$ need to form a basis in the appropriate functional space, i.e., in particular, $\{\phi_j(x)\}_{1\le j\le N}$ need to be linearly independent; the necessary and sufficient condition that $\{\phi_j(x)\}_{1\le j\le N}$ are to satisfy is $\det[\phi_j(z_i)]\neq 0$. As $\{\phi_j(x)\}_{1\le j\le N}$ is a basis, expansion~\eqref{eq:u_N-def} is equivalent to acting on the sought function, $u(x)$, by an interpolatory projection operator, $\pi_N$, that projects $u(x)$ onto the span of $\{\phi_j(x)\}_{1\le j\le N}$. This operator is defined as $\pi_N\circ u\triangleq\sum_{j=1}^N u_{j,N} \phi_j(x)$ with $\|\pi_N\|_\infty \triangleq\max_{0\le x\le A}\sum_{j=1}^N|\phi_j(x)|\ge1$.

By design, the described method is most accurate at the collocation nodes, $\{z_j\}_{1\le j\le N}$, since it is at these points that the residual is zero. For an arbitrary point $x\not\in\{z_j\}_{1\le j\le N}$, the unknown function, $u(x)$, can be evaluated as
\begin{align}
\begin{aligned}\label{eq:int-eqn-iter-sol}
\widetilde{u}_N(x)
&=
\upsilon(x)+\int_0^A{K}_\infty(x,y)\,u_N(y)\,dy\\
&=
\upsilon(x)+\sum_{j=1}^Nu_{j,N}\int_0^A{K}_\infty(x,y)\,\phi_j(y)\,dy.
\end{aligned}
\end{align}
This technique is known as the iterated projection solution; see, e.g.,~\cite[Section~12.3]{Atkinson+Han:Book09}; note that $\widetilde{u}_N(z_j)=u_N(z_j)=u_{j,N}$, $1\le j\le N$.

We now consider the question of the method's accuracy and rate of convergence. To that end, it is apparent that the choice of $\{\phi_j(x)\}_{1\le j\le N}$ must play a critical role. This is, in fact, the case, as may be concluded from, e.g.,~\cite[Theorem~12.1.12,~p.~479]{Atkinson+Han:Book09}. Specifically, using $\|u-\widetilde{u}_N\|_\infty$ as a sensible measure of the method's error, and applying~\cite[Formula~12.3.21,~p.~499]{Atkinson+Han:Book09}, we obtain
\begin{align}\label{eq:err-bound-1}
\|u-\widetilde{u}_N\|_\infty
&\le
\|(\mathbb{I}-\mathcal{K}_\infty)^{-1}\|_\infty\|\mathcal{K}_\infty\circ(\mathbb{I}-\pi_N)\circ u\|_\infty\le
\|(\mathbb{I}-\mathcal{K}_\infty)^{-1}\|_\infty\|\mathcal{K}_\infty\|_\infty\|(\mathbb{I}-\pi_N)\circ u\|_\infty,
\end{align}
whence one can see that the method's error is determined by $\|(\mathbb{I}-\mathcal{K}_\infty)^{-1}\|_\infty$ and by $\|(\mathbb{I}-\pi_N)\circ u\|_\infty$; the latter is the interpolation error and can be found for each particular choice of $\pi_N$, which requires choosing the basis $\{\phi_j(x)\}_{1\le j\le N}$ and the collocation nodes $\{z_j\}_{1\le j\le N}$. The bigger problem, therefore, is to upperbound $\|(\mathbb{I}-\mathcal{K}_\infty)^{-1}\|_\infty$. To that end, the standard result
\begin{align*}
\|(\mathbb{I}-\mathcal{K}_\infty)^{-1}\|_\infty
&\le
\dfrac{1}{1-\|\mathcal{K}_\infty\|_\infty}
\end{align*}
is applicable, since $\|\mathcal{K}_\infty\|_\infty<1$. However, it is well-known that this is often a very crude inequality, and it may not be practical to use it in~\eqref{eq:err-bound-1} to upperbound $\|u-\widetilde{u}_N\|_\infty$. Since in our particular case $\mathcal{K}_\infty$ is the transition probability kernel of a stationary Markov process, a tighter (in fact, exact) upperbound on $\|(\mathbb{I}-\mathcal{K}_\infty)^{-1}\|_\infty$ is possible to obtain. We now state the corresponding result first established in~\cite[Lemma~3.1]{Polunchenko+etal:SA2014}; see also~\cite{Polunchenko+etal:MIPT2013}.
\begin{lemma}\label{lem:ARL-inv-bound}
$\|\,(\mathbb{I}-\mathcal{K}_\infty)^{-1}\|_\infty=\|\,\ell\,\|_\infty$.
\end{lemma}

With this lemma one can upperbound $\|u-\widetilde{u}_N\|_\infty$ rather tightly. Specifically, from~\eqref{eq:err-bound-1}, $\|\mathcal{K}_\infty\|_\infty<1$, and Lemma~\ref{lem:ARL-inv-bound}, we obtain $\|u-\widetilde{u}_N\|_\infty<\|\,\ell\,\|_\infty\|(\mathbb{I}-\pi_N)\circ u\|_\infty$, where the inequality is strict because $\|\mathcal{K}_\infty\|_\infty$ is strictly less than $1$. The only question now is the interpolation error $\|(\mathbb{I}-\pi_N)\circ u\|_\infty$, which is determined by the choice of $\pi_N$. To that end, for reasons to be explained below, we propose to seek the solution, $u(x)$, within the piecewise linear polynomial space. Specifically, given a positive integer $N\ge2$, let $\Pi_N\colon 0\triangleq x_0<x_1<\ldots<x_{N-1}\triangleq A$ denote a partition of the interval $[0,A]$, and for $j=1,\ldots,N-1$ set $I_j^N\triangleq(x_{j-1}, x_j)$, $h_j\triangleq x_j-x_{j-1}(>0)$, and $h\triangleq h(N)=\max_{1\le j\le N-1} h_j$; assume also that $h\to0$, as $N\to\infty$. Next, set $z_j=x_{j-1}$, $1\le j\le N$ and choose the basis $\{\phi_j(x)\}_{1\le j\le N}$ of the ``hat'' functions
\begin{empheq}[%
    left={%
        \phi_j(x)=%
    \empheqlbrace}]{align}\nonumber
&\frac{x-x_{j - 2}}{h_{j-1}},\quad\text{if $x\in I_{j-1}^N, j>1$;}\label{eq:lin-basis}\\
&\frac{x_j-x}{h_{j}},\quad\text{if $x\in I_{j}^N, j<N$;}\\\nonumber
&0,\quad\text{otherwise},
\end{empheq}
where $1\le j\le N$; cf.~\cite{Polunchenko+etal:MIPT2013,Polunchenko+etal:SA2014}.

For this choice of the functional basis $\{\phi_j(x)\}_{1\le j\le N}$ it is known~\cite[Formula~3.2.9,~p.~124]{Atkinson+Han:Book09} that $\|(\mathbb{I}-\pi_N)\circ u\|_\infty\le\|\,u_{xx}\,\|_\infty h^2/8$, where $u_{xx}\triangleq \partial^2 u(x)/\partial x^2$. Hence, the method's rate of convergence is quadratic and $\|u-\tilde{u}_N\|<\|\,\ell\,\|_\infty\|\,u_{xx}\,\|_\infty h^2/8$. This result can now be ``tailored'' to the equations of interest, namely, to equations~\eqref{eq:ARL-int-eqn} and~\eqref{eq:psi+delta-int-eqn}.
\begin{theorem}\label{thm:error-bound-ARL}
Given $N\ge2$ sufficiently large
\begin{align*}
\|\ell-\widetilde{\ell}_N\|_\infty
&<
\|\,\ell\,\|_\infty\|\,\ell_{xx}\,\|_\infty\frac{h^2}{8},\;\text{where}\; \ell_{xx}\triangleq\dfrac{\partial^2}{\partial x^2}\ell(x,A);
\end{align*}
note that the inequality is strict.
\end{theorem}
\begin{theorem}\label{thm:error-bound-lwrbnd-numerator}
Given $N\ge2$ sufficiently large
\begin{align*}
\|\Xi-\widetilde{\Xi}_N\|_\infty
&<
\|\,\ell\,\|_\infty\|\,\Xi_{xx}\,\|_\infty\frac{h^2}{8},\;\text{where}\; \Xi_{xx}\triangleq\dfrac{\partial^2}{\partial x^2}\Xi(x,A);
\end{align*}
note that the inequality is strict.
\end{theorem}

The error bound given in Theorem~\ref{thm:error-bound-ARL} was first obtained in~\cite[Theorem~3.1]{Polunchenko+etal:SA2014}. Note that the bound is proportional to the magnitude of the solution, i.e., to $\ell(x,A)\triangleq\ARL(\mathcal{S}_A^{r=x})$, which can be large. Worse yet, the bound is also proportional to the detection threshold squared lurking in numerator of $h^2$ (for simplicity assume that $h=A/N$). Since $\ell(x,A)\approx A/\xi-x$ with $\xi\in(0,1)$, one can roughly set $A\approx\ell(x,A)$ and conclude that the error bound is roughly proportional to $\ell^3(x,A)$, i.e., to the magnitude of the solution cubed. This may seem to drastically offset the second power of $N$ buried in the denominator of $h^2$. However, as was already argued and confirmed experimentally in~\cite{Polunchenko+etal:MIPT2013,Polunchenko+etal:SA2014}, this does not happen. The reason is the (almost) linearity of $\ell(x,A)\triangleq\ARL(\mathcal{S}_A^{r=x})$ with respect to the headstart $r=x$, as evident from the approximation~\eqref{eq:SRr-ARL-approx}. Specifically, $\ell(x,A)\approx A/\xi-x$, and therefore
\begin{align*}
\frac{\partial^2}{\partial x^2}
\ell(x,A)
&\approx 0,\;\text{at least for}\; x\in[0,A].
\end{align*}
This makes the error bound given in Theorem~\ref{thm:error-bound-ARL} extremely close to zero, even for relatively small $N$. Consequently, $\ell(x,A)$ can be computed rather accurately without requiring $N$ to be large. This is one of the reasons to use the above piecewise linear basis~\eqref{eq:lin-basis}.

The error bound given in Theorem~\ref{thm:error-bound-lwrbnd-numerator} is not as close to zero because, unlike $\ell(x,A)$, the the function $\Xi(x,A)$ is not linear in $x$. Nevertheless, as will be shown experimentally in the next section, the method's accuracy and robustness for $\Xi(x,A)$ are substantially better than those of the method proposed in~\cite{Tartakovsky+etal:IWSM2009,Moustakides+etal:CommStat09,Moustakides+etal:SS11}.

There is one more purpose that the change-of-measure identity, $(1+x)\,{K}_0(x,y)=y\,{K}_\infty(x,y)$, serves: it is used to compute the matrix~\eqref{eq:def-matrix-K} required to implement the proposed numerical method. Specifically, due to the change-of-measure identity, the integrals involved in~\eqref{eq:def-matrix-K} can be computed {\it exactly}: using~\eqref{eq:def-matrix-K} and~\eqref{eq:lin-basis}, and recalling that $z_j=x_{j-1}$, $1\le j\le N$, the corresponding formula is
\begin{align}
\begin{aligned}\label{eq:matrix-K-new-method}
(\boldsymbol{K}_\infty)_{i,j}
&=
\int_0^A{K}_\infty(x_i,y)\,\phi_j(y)\,dy\\
&=
\dfrac{1}{h_{j-1}}\left\{(1+x_i)\left[{P}_0^{\LR}\left(\dfrac{x_{j-1}}{1+x_i}\right)-{P}_0^{\LR}\left(\dfrac{x_{j-2}}{1+x_i}\right)\right]\right.-\\
&\qquad\qquad\qquad\qquad\qquad x_{j-2}\left.\left[{P}_\infty^{\LR}\left(\dfrac{x_{j-1}}{1+x_i}\right)-{P}_\infty^{\LR}\left(\dfrac{x_{j-2}}{1+x_i}\right)\right]\right\}\indicator{j>1}+\\
&\qquad\dfrac{1}{h_{j}}\left\{x_j\left[{P}_\infty^{\LR}\left(\dfrac{x_{j}}{1+x_i}\right)-{P}_\infty^{\LR}\left(\dfrac{x_{j-1}}{1+x_i}\right)\right]\right.-\\
&\qquad\qquad\qquad\qquad\qquad (1+x_i)\left.\left[{P}_0^{\LR}\left(\dfrac{x_{j}}{1+x_i}\right)-{P}_0^{\LR}\left(\dfrac{x_{j-1}}{1+x_i}\right)\right]\right\}\indicator{j<N}
\end{aligned}
\end{align}
for $1\le i,j\le N$; cf.~\cite{Polunchenko+etal:MIPT2013,Polunchenko+etal:SA2014}.

To wrap this subsection, note that the proposed method is a numerical framework that can also be used to assess the accuracy of the popular Markov chain approach, introduced in~\cite{Brook+Evans:B1972}, and later extended, e.g., in~\cite{Woodall:T1983}. To this end, as noted in~\cite{Champ+Rigdon:CommStat1991}, the Markov chain approach is equivalent to the integral-equations approach if the integral is approximated via the product midpoint rule. This, in turn, is equivalent to choosing the basis functions, $\{\phi_j(x)\}_{1\le j\le N}$, as piecewise constants on $\Pi_{N+1}$, i.e., $\phi_j(x)=\indicator{x\in I_j^{N+1}}$, and equating the residual to zero at the midpoints of the intervals $I_j^{N+1}$, i.e., setting $z_j=(x_{j-1}+x_j)/2$, $1\le j\le N$. In this case the $(i,j)$-th element of the matrix $\boldsymbol{K}$ defined by~\eqref{eq:def-matrix-K} is
\begin{align*}
(\boldsymbol{K}_\infty)_{i,j}
&=
P_\infty^{\LR}\left(\dfrac{x_j}{1+z_i}\right)
-
P_\infty^{\LR}\left(\dfrac{x_{j-1}}{1+z_i}\right),\;1\le i,j\le N;
\end{align*}
cf.~\cite{Tartakovsky+etal:IWSM2009,Moustakides+etal:CommStat09}. It can be shown (see, e.g.,~\cite{Kryloff+Bogoliubov:CRA-URSS1929} or~\cite[pp.~130--135]{Kantorovich+Krylov:Book1958}) that this approach exhibits a superconvergence effect: the rate is also quadratic, even though the interpolation is based on polynomials of degree zero (i.e., constants, or step functions). However, in spite of the superconvergence and the much simpler matrix $\boldsymbol{K}$, the constant in front of $h^2$ in the corresponding error bound is large (larger than that for the ``hat'' functions). As a result, the partition size required by this method ends up being substantial. In fact, this method was employed, e.g., in~\cite{Tartakovsky+etal:IWSM2009,Moustakides+etal:CommStat09}, to compare the CUSUM chart and the original SR procedure, and the partition size used consisted of thousands of points to ensure reasonable accuracy. The comparison of this method and the proposed method for $\ell(x,A)$ performed in~\cite{Polunchenko+etal:MIPT2013,Polunchenko+etal:SA2014} confirmed that the new method is superior. In the next section we will offer the same comparison but for $\STADD(\mathcal{S}_A^r)$, and confirm that the new method is superior in this case as well.

%-------------------------------------------------------------------------------------------------%
\section{A case study}
\label{sec:case-study}

As an illustration of the proposed numerical method at work, consider a scenario where the observations, $\{X_n\}_{n\ge1}$, are independent Gaussian with mean zero pre-change and mean $\theta\neq0$ (known) post-change; the variance is 1 and does not change. Formally, the pre- and post-change distribution densities in this case are
\begin{align}\label{eq:gaussian-scenario}
f(x)
&=
\dfrac{1}{\sqrt{2\pi}}\exp\left\{-\dfrac{x^2}{2}
\right\}\;\text{and}\;
g(x)=\dfrac{1}{\sqrt{2\pi}}\exp\left\{-\dfrac{(x-\theta)^2}{2}\right\},
\end{align}
respectively, where $x\in\mathbb{R}$ and $\theta\neq0$. The corresponding ``instantaneous'' LR for the $n$-th data point, $X_n$, can be seen to be
\begin{align*}
\LR_n
&\triangleq
\dfrac{g(X_n)}{f(X_n)}
=\exp\left\{\theta X_n-\frac{\theta^2}{2}\right\},\;n\ge1,
\end{align*}
and, therefore, for each $n\ge1$ its distribution is log-normal with mean $-\theta^2/2$ and variance $\theta^2$ under measure $\Pr_\infty$, and with mean $\theta^2/2$ and variance $\theta^2$ under measure $\Pr_0$. Consequently, one can use~\eqref{eq:matrix-K-new-method} to find the matrix $\boldsymbol{K}$ required to implement the proposed method. Also, since in this case
\begin{align*}
{K}_\infty(x,y)
&=
\dfrac{1}{y\sqrt{2\pi\theta^2}}\exp\left\{-\dfrac{1}{2\theta^2}\left(\log\dfrac{y}{1+x}+\dfrac{\theta^2}{2}\right)^2\right\}\indicator{y/(1+x)\ge0},
\end{align*}
one can see that it is indifferent whether $\theta<0$ or $\theta>0$. We, therefore, without loss of generality, will consider only the former case, i.e., assume from now on that $\theta>0$.

\begin{remark}
It is not necessary to find the formula for ${K}_0(x,y)$ since for the proposed method it is sufficient to know $K_\infty(x,y)$ only. Yet, if it were necessary to have an explicit expression for ${K}_0(x,y)$, it would be easy to obtain it from the above formula for ${K}_\infty(x,y)$ and the identity $(1+x)\,{K}_0(x,y)=y\,{K}_\infty(x,y)$ established in Subsection~\eqref{ssec:integral-equations} using the change-of-measure identity~\eqref{eq:change-of-measure-identity}.
\end{remark}

We now employ the proposed numerical method and its predecessor offered and applied in~\cite{Tartakovsky+etal:IWSM2009,Moustakides+etal:CommStat09,Moustakides+etal:SS11} to evaluate the performance of the GSR procedure~\eqref{SR-rst}-\eqref{SR-rstat} for the Gaussian scenario~\eqref{eq:gaussian-scenario}. Our intent is to assess and compare the quality of each of the two methods. For the ARL to false alarm, this task was already accomplished in~\cite{Polunchenko+etal:MIPT2013,Polunchenko+etal:SA2014}, and, as expected from the discussion in the end of Subsection~\ref{ssec:numerical-method+accuracy}, the new method was confirmed to be rather accurate and robust, far surpassing its predecessor. We, therefore, shall devise the two methods to compute the STADD only. More specifically, as in~\cite{Polunchenko+etal:MIPT2013,Polunchenko+etal:SA2014}, we will examine the sensitivity of the STADD computed by each of the two methods using definition~\eqref{eq:STADD-gen-def} with respect to three factors:\begin{inparaenum}[\itshape a)]\item partition fineness (rough vs. fine), \item change magnitude (faint vs. contrast), and \item value of the ARL to false alarm (low vs. high)\end{inparaenum}.

As was mentioned in Subsection~\ref{ssec:numerical-method+accuracy}, the accuracy of the proposed method is determined by the accuracy of the underlying piecewise linear polynomial interpolation with basis~\eqref{eq:lin-basis}; see Theorems~\ref{thm:error-bound-ARL} and~\ref{thm:error-bound-lwrbnd-numerator}. Since the interpolation basis~\eqref{eq:lin-basis} is fixed, the corresponding interpolation error is dependent upon how the interval of interpolation (i.e., $[0,A]$) is partitioned. To that end, recall that if the interval is partitioned into non-overlapping subintervals joint at the Chebyshev abscissas (i.e., roots of the Chebyshev polynomials of the first kind), then the corresponding interpolation error is the smallest possible; see, e.g.,~\cite[Section~8.3]{Burden+Faires:Book2011}. Thus, to improve the overall accuracy of the method, we follow~\cite{Polunchenko+etal:MIPT2013,Polunchenko+etal:SA2014} and partition the interval $[0,A]$ into $N-1$, $N\ge2$, non-overlapping subintervals $I_j^N\triangleq(x_{j-1},x_j)$, $1\le j\le N-1$, joint at the shifted Chebyshev abscissas
\begin{align*}
x_{N-j}
&=
\frac{A}{2}\left\{1+\cos\left[(2j - 1)\dfrac{\pi}{2N}\right]/\cos\left(\dfrac{\pi}{2N}\right)\right\},\;1\le j\le N,
\end{align*}
where the shift is to make sure that $x_0=0$ and $x_{N-1}=A$; these points are also the collocation nodes $z_j$, i.e., $z_j=x_{j-1}$, $1\le j\le N$. Using $h\triangleq\max_{1\le j\le N-1}h_j$ with $h_j\triangleq x_{j}-x_{j-1}$ as a measure the partition fineness it can be shown that in this case
\begin{align*}
h_j
&=
A \tan \left(\frac{\pi }{2 N}\right) \sin \left(\frac{\pi  j}{N}\right),\; 1 \le j \le N-1,
\end{align*}
whence $h\triangleq\max_{1 \le j \le N-1} h_j=h_{\lfloor N / 2 \rfloor}$ with $\lfloor x\rfloor$ being the floor function; note that $h$ is roughly of order $A/N$ for sufficiently large $N$. For a reason explained shortly it is convenient to set the partition size, $N$, to be of the form $N=2^{j}$ for $j=1,2,\ldots$. By varying $j$, the partition can then be made more rough (small $j$) or more fine (large $j$). We will consider $j=1,2,\ldots,12$.

For the Gaussian scenario~\eqref{eq:gaussian-scenario} the magnitude of the change is represented by $\theta$. We will consider $\theta=0.01, 0.1, 0.5$ and $1.0$, which correspond to a very faint, small, moderate, and contrast change, respectively. For the ARL to false alarm ($\ARL(\mathcal{S}_A^r)=\gamma$) we will consider levels $\gamma=10^2, 10^3, 10^4$ and even $10^5$, although the latter is an extreme case and unlikely to be practical.

To measure the accuracy and rate of convergence of either of the two methods we will rely on the standard Richardson extrapolation technique: if $u_{2N}$, $u_N$ and $u_{N/2}$ are the solutions (of the corresponding integral equation) obtained assuming the partition size is $2N$, $N$ and $N/2$, respectively, then the rate of convergence, $c$, can be estimated as
\begin{align*}
2^{-c}
&\approx
\dfrac{\|u_{2N}-u_N\|_\infty}{\|u_{N}-u_{N/2}\|_\infty}\;\text{so that}\;
c\approx-\log_2\dfrac{\|u_{2N}-u_N\|_\infty}{\|u_{N}-u_{N/2}\|_\infty},
\end{align*}
and the actual error, $\|u-u_N\|_\infty$, can be estimated as $\|u-u_N\|_\infty\approx2^{-c}\|u_{N}-u_{N/2}\|_\infty$. This is why it is convenient to make the partition size, $N$, to be of the form $N=2^j$ for $j=1,2,\ldots$. As we mentioned before, we will consider $N=2^j$ for $j=1,2,\ldots,12$.

Both methods were implemented and tested in MATLAB. Tables~\ref{tab:STADD_vs_A_N__theta__0_01},~\ref{tab:STADD_vs_A_N__theta__0_10},~\ref{tab:STADD_vs_A_N__theta__0_50} and~\ref{tab:STADD_vs_A_N__theta__1_00} present the obtained results. Each table is for a specific change magnitude ($\theta=0.01,0.1,0.5$, and $1.0$), and reports the results in four two-column blocks, one for each of the selected values of the ARL to false alarm ($\gamma=10^2,10^3,10^4$, and $10^5$). Within each of the four two-column blocks, the left column reports the obtained values of the STADD for $r=0$, i.e., $\STADD(\mathcal{S}_A)$, and the right column reports the corresponding empirical estimate of the convergence rate (if it is available); the performance at an $r$ different from any of the collocation nodes $z_j$ can be computed using the iterated solution~\eqref{eq:int-eqn-iter-sol}. The numbers in brackets are the values of the STADD outputted by the predecessor method; in particular, \verb"[NaN]" indicates that the method failed. From the presented results one can conclude that:\begin{inparaenum}[\itshape a)]\item as expected, the convergence rate of the new method is, in fact, quadratic and \item it is achieved much quicker than for the predecessor method, for a broad range of values of the ARL to false alarm and change magnitudes\end{inparaenum}. Hence, the method is not only more accurate, but is also more robust. To boot, we reiterate that the new method is also more efficient as it can compute both the ARL to false alarm and the STADD {\em simultaneously}.
\begin{table}
    \centering
    \scalebox{0.65}{
    \begin{tabularx}{1.5\textwidth}{c *{9}{Y}}
    \toprule

     & \multicolumn{2}{c}{$\gamma=10^2$ ($A=99.2$)}
     & \multicolumn{2}{c}{$\gamma=10^3$ ($A=994.2$)}
     & \multicolumn{2}{c}{$\gamma=10^4$ ($A=9,941.9$)}
     & \multicolumn{2}{c}{$\gamma=10^5$ ($A=99,419.0$)}\\
    \cmidrule(lr){2-3} \cmidrule(lr){4-5} \cmidrule(lr){6-7} \cmidrule(lr){8-9}
          $N$ & $\STADD(\mathcal{S}_A)$ & Rate & $\STADD(\mathcal{S}_A)$ & Rate & $\STADD(\mathcal{S}_A)$ & Rate & $\STADD(\mathcal{S}_A)$ & Rate\\
    \midrule
        \multirow{2}{*}{2}&2.17402&-&2.84308&-&2.97606&-&2.99044&-\\
            &[NaN]&-&[NaN]&-&[NaN]&-&[NaN]&-\\
            \midrule
        \multirow{2}{*}{4}&34.11371&1.80193&169.48315&0.50123&244.24921&-0.76361&254.47215&-0.97507\\
            &[NaN]&[NaN]&[NaN]&[NaN]&[NaN]&[NaN]&[NaN]&[NaN]\\
            \midrule
        \multirow{2}{*}{8}&43.27374&1.22992&287.21487&0.49899&653.86714&-0.45861&748.82025&-0.92611\\
            &[NaN]&[NaN]&[NaN]&[NaN]&[NaN]&[NaN]&[NaN]&[NaN]\\
            \midrule
        \multirow{2}{*}{16}&47.17904&1.04689&370.52193&0.54205&1,216.76987&-0.31204&1,688.15411&-0.82836\\
            &[NaN]&[NaN]&[NaN]&[NaN]&[NaN]&[NaN]&[NaN]&[NaN]\\
            \midrule
        \multirow{2}{*}{32}&49.06923&1.05343&427.7369&0.76084&1,915.59076&-0.19302&3,356.09536&-0.71689\\
            &[NaN]&[NaN]&[NaN]&[NaN]&[NaN]&[NaN]&[NaN]&[NaN]\\
            \midrule
        \multirow{2}{*}{64}&49.97997&1.5395&461.50242&1.08494&2,714.44934&0.18883&6,097.57295&-0.57304\\
            &[53.76583]&-&[NaN]&[NaN]&[NaN]&[NaN]&[NaN]&[NaN]\\
            \midrule
        \multirow{2}{*}{128}&50.29327&2.44785&477.41991&1.50683&3,415.30374&0.88027&10,175.95755&-0.16839\\
            &[48.89458]&[1.52049]&[NaN]&[NaN]&[NaN]&[NaN]&[NaN]&[NaN]\\
            \midrule
        \multirow{2}{*}{256}&50.3507&1.92387&483.02102&1.86544&3,796.05373&1.63342&14,759.26034&0.59232\\
            &[50.59254]&[2.8191]&[NaN]&[NaN]&[NaN]&[NaN]&[NaN]&[NaN]\\
            \midrule
        \multirow{2}{*}{512}&50.36583&1.99945&484.55819&2.00587&3,918.77793&1.94218&17,799.2511&1.44963\\
            &[50.35194]&[3.67929]&[487.19086]&-&[NaN]&[NaN]&[NaN]&[NaN]\\
            \midrule
        \multirow{2}{*}{1024}&50.36962&2.00029&484.94092&1.99957&3,950.71356&1.99186&18,912.23803&1.88795\\
            &[50.37072]&[3.1533]&[484.92745]&[4.89991]&[NaN]&[NaN]&[NaN]&[NaN]\\
            \midrule
        \multirow{2}{*}{2048}&50.37056&1.99993&485.03663&2.0&3,958.74262&1.99861&19,212.95701&1.97715\\
            &[50.36861]&[-0.17441]&[485.00326]&[0.96341]&[NaN]&[NaN]&[NaN]&[NaN]\\
            \midrule
        \multirow{2}{*}{4096}&50.3708&-&485.06056&-&3,960.75182&-&19,289.33685&-\\
            &[50.371]&-&[485.04214]&-&[NaN]&-&[NaN]&-\\
    \bottomrule
    \end{tabularx}
    } % endof \scalebox
    \caption{Results of accuracy and convergence analysis for $\STADD(\mathcal{S}_A)$ for $\theta=0.01$.}
    \label{tab:STADD_vs_A_N__theta__0_01}
\end{table}
\begin{table}
    \centering
    \scalebox{0.65}{
    \begin{tabularx}{1.5\textwidth}{c *{9}{Y}}
    \toprule
     & \multicolumn{2}{c}{$\gamma=10^2$ ($A=94.34$)}
     & \multicolumn{2}{c}{$\gamma=10^3$ ($A=943.41$)}
     & \multicolumn{2}{c}{$\gamma=10^4$ ($A=9,434.08$)}
     & \multicolumn{2}{c}{$\gamma=10^5$ ($A=94,340.5$)}\\
    \cmidrule(lr){2-3} \cmidrule(lr){4-5} \cmidrule(lr){6-7} \cmidrule(lr){8-9}
          $N$ & $\STADD(\mathcal{S}_A)$ & Rate & $\STADD(\mathcal{S}_A)$ & Rate & $\STADD(\mathcal{S}_A)$ & Rate & $\STADD(\mathcal{S}_A)$ & Rate\\
    \midrule
        \multirow{2}{*}{2}&2.76144&-&2.9058&-&2.92154&-&2.92313&-\\
            &[24.585]&-&[236.8525]&-&[2,359.52]&-&[23,586.12504]&-\\
            \midrule
        \multirow{2}{*}{4}&19.57682&0.67078&28.49977&-0.65615&29.7676&-0.88402&29.89945&-0.90832\\
            &[12.7925]&[1.00003]&[118.92625]&[1.0]&[1,180.26]&[1.0]&[11,793.56254]&[1.0]\\
            \midrule
        \multirow{2}{*}{8}&30.13966&0.72401&68.83245&-0.31003&79.31229&-0.82197&80.53033&-0.89334\\
            &[6.8964]&[1.04547]&[59.96313]&[1.0]&[590.63]&[1.0]&[5,897.2813]&[1.0]\\
            \midrule
        \multirow{2}{*}{16}&36.53453&1.27722&118.83417&0.21564&166.89827&-0.48788&174.57576&-0.68236\\
            &[4.03982]&[-3.28195]&[30.48157]&[1.0]&[295.81501]&[1.0]&[2,949.14068]&[1.0]\\
            \midrule
        \multirow{2}{*}{32}&39.17299&1.86583&161.89383&0.95142&289.72703&0.05887&325.4964&-0.36726\\
            &[31.82495]&[1.72421]&[15.74081]&[1.0002]&[148.40751]&[1.0]&[1,475.07037]&[1.0]\\
            \midrule
        \multirow{2}{*}{64}&39.89689&1.99622&184.16095&1.68328&407.64494&0.71091&520.16925&0.02383\\
            &[40.23452]&[6.50986]&[8.37144]&[1.18069]&[74.70377]&[1.0]&[738.03522]&[1.0]\\
            \midrule
        \multirow{2}{*}{128}&40.07834&1.99758&191.09434&1.94319&479.68504&1.42376&711.65248&0.50423\\
            &[40.14224]&[2.44728]&[5.12052]&[-5.52996]&[37.85191]&[1.00002]&[369.51765]&[1.0]\\
            \midrule
        \multirow{2}{*}{256}&40.12377&1.99939&192.8973&1.98713&506.5372&1.86034&846.65517&1.08119\\
            &[40.12532]&[0.40575]&[155.32721]&[1.97711]&[19.42619]&[1.00099]&[185.25886]&[1.0]\\
            \midrule
        \multirow{2}{*}{512}&40.13514&1.99985&193.35208&1.99680&513.93259&1.97013&910.46281&1.6616\\
            &[40.13809]&[4.42044]&[193.47941]&[10.10018]&[10.21963]&[1.48965]&[93.12949]&[1.0]\\
            \midrule
        \multirow{2}{*}{1024}&40.13798&1.99996&193.46603&1.99920&515.82011&1.99265&930.6316&1.92646\\
            &[40.13869]&[1.74523]&[193.51416]&[2.81961]&[6.94119]&[-6.92857]&[47.0649]&[1.00006]\\
            \midrule
        \multirow{2}{*}{2048}&40.13869&1.99999&193.49453&1.99980&516.2944&1.99817&935.93748&1.98293\\
            &[40.13887]&[2.0]&[193.50924]&[-0.24306]&[406.30949]&[1.86125]&[24.0336]&[1.00371]\\
            \midrule
        \multirow{2}{*}{4096}&40.13887&-&193.50165&-&516.41313&-&937.27974&-\\
            &[40.13891]&-&[193.50341]&-&[516.23069]&-&[12.54751]&-\\
    \bottomrule
    \end{tabularx}
    } % endof \scalebox
    \caption{Results of accuracy and convergence analysis for $\STADD(\mathcal{S}_A)$ for $\theta=0.1$.}
    \label{tab:STADD_vs_A_N__theta__0_10}
\end{table}
\begin{table}
    \centering
    \scalebox{0.65}{
    \begin{tabularx}{1.5\textwidth}{c *{9}{Y}}
    \toprule
     & \multicolumn{2}{c}{$\gamma=10^2$ ($A=74.76$)}
     & \multicolumn{2}{c}{$\gamma=10^3$ ($A=747.62$)}
     & \multicolumn{2}{c}{$\gamma=10^4$ ($A=7,476.15$)}
     & \multicolumn{2}{c}{$\gamma=10^5$ ($A=74,761.5$)}\\
    \cmidrule(lr){2-3} \cmidrule(lr){4-5} \cmidrule(lr){6-7} \cmidrule(lr){8-9}
          $N$ & $\STADD(\mathcal{S}_A)$ & Rate & $\STADD(\mathcal{S}_A)$ & Rate & $\STADD(\mathcal{S}_A)$ & Rate & $\STADD(\mathcal{S}_A)$ & Rate\\
    \midrule
        \multirow{2}{*}{2}&2.6205&-&2.66516&-&2.66976&-&2.67022&-\\
            &[23.26223]&-&[218.41569]&-&[2,170.03728]&-&[21,686.28053]&-\\
            \midrule
        \multirow{2}{*}{4}&6.63261&0.10081&7.35814&-0.51049&7.43947&-0.58250&7.44771&-0.58983\\
            &[15.9271]&[1.13548]&[134.21152]&[0.73966]&[1,318.04091]&[0.71719]&[13,156.46809]&[0.7151]\\
            \midrule
        \multirow{2}{*}{8}&10.37395&1.26845&14.04345&-0.07021&14.58186&-0.32169&14.63816&-0.3489\\
            &[12.58828]&[2.28418]&[83.78332]&[0.74448]&[799.78634]&[0.69961]&[7,960.42708]&[0.69591]\\
            \midrule
        \multirow{2}{*}{16}&11.927&1.89783&21.06218&0.69112&23.50837&-0.05265&23.79588&-0.15479\\
            &[11.90281]&[0.82232]&[53.68363]&[0.82221]&[480.67633]&[0.71915]&[4,752.81859]&[0.71248]\\
            \midrule
        \multirow{2}{*}{32}&12.34375&1.96487&25.40939&1.60408&32.76669&0.336&33.99079&-0.02254\\
            &[12.29046]&[1.23906]&[36.65994]&[1.04091]&[286.83166]&[0.75063]&[2,795.30806]&[0.73797]\\
            \midrule
        \multirow{2}{*}{64}&12.45051&1.99174&26.83938&1.90115&40.10145&1.09236&44.34624&0.1498\\
            &[12.45469]&[2.67661]&[28.38608]&[1.87939]&[171.62094]&[0.79145]&[1,621.62076]&[0.76513]\\
            \midrule
        \multirow{2}{*}{128}&12.47735&1.99794&27.22223&1.97525&43.54141&1.77992&53.68039&0.58381\\
            &[12.48038]&[2.50367]&[26.13725]&[1.6805]&[105.05666]&[0.85487]&[931.02473]&[0.79149]\\
            \midrule
        \multirow{2}{*}{256}&12.48407&1.99949&27.3196&1.99376&44.54312&1.93619&59.90811&1.39801\\
            &[12.48491]&[2.11147]&[26.83883]&[0.59481]&[68.25222]&[1.001]&[532.03476]&[0.81722]\\
            \midrule
        \multirow{2}{*}{512}&12.48575&1.99987&27.34404&1.99844&44.80487&1.98475&62.27124&1.86278\\
            &[12.48596]&[1.96975]&[27.30336]&[3.49023]&[49.8628]&[1.50297]&[305.59341]&[0.84533]\\
            \midrule
        \multirow{2}{*}{1024}&12.48617&1.99997&27.35016&1.99961&44.871&1.9962&62.92097&1.96207\\
            &[12.48622]&[2.00126]&[27.3447]&[2.88607]&[43.37452]&[4.52839]&[179.55993]&[0.88647]\\
            \midrule
        \multirow{2}{*}{2048}&12.48628&1.99999&27.35169&1.9999&44.88758&1.99905&63.08773&1.99058\\
            &[12.48629]&[2.0]&[27.35029]&[1.9689]&[43.65568]&[-1.95213]&[111.38368]&[0.97978]\\
            \midrule
        \multirow{2}{*}{4096}&12.4863&-&27.35207&-&44.89173&-&63.12969&-\\
            &[12.48631]&-&[27.35172]&-&[44.7436]&-&[76.81434]&-\\
    \bottomrule
    \end{tabularx}
    } % endof \scalebox
    \caption{Results of accuracy and convergence analysis for $\STADD(\mathcal{S}_A)$ for $\theta=0.5$.}
    \label{tab:STADD_vs_A_N__theta__0_50}
\end{table}
\begin{table}
    \centering
    \scalebox{0.65}{
    \begin{tabularx}{1.5\textwidth}{c *{9}{Y}}
    \toprule
     & \multicolumn{2}{c}{$\gamma=10^2$ ($A=56.0$)}
     & \multicolumn{2}{c}{$\gamma=10^3$ ($A=560.0$)}
     & \multicolumn{2}{c}{$\gamma=10^4$ ($A=5,603.5$)}
     & \multicolumn{2}{c}{$\gamma=10^5$ ($A=56,037.0$)}\\
    \cmidrule(lr){2-3} \cmidrule(lr){4-5} \cmidrule(lr){6-7} \cmidrule(lr){8-9}
          $N$ & $\STADD(\mathcal{S}_A)$ & Rate & $\STADD(\mathcal{S}_A)$ & Rate & $\STADD(\mathcal{S}_A)$ & Rate & $\STADD(\mathcal{S}_A)$ & Rate\\
    \midrule
        \multirow{2}{*}{2}&2.40791&-&2.44232&-&2.44582&-&2.44617&-\\
            &[17.26312]&-&[162.22888]&-&[1,612.83039]&-&[16,118.40774]&-\\
            \midrule
        \multirow{2}{*}{4}&3.65841666&-0.01705&3.91567&-0.57029&3.94348&-0.62773&3.94628&-0.63353\\
            &[11.56644]&[0.74976]&[101.10593]&[0.64531]&[996.98326]&[0.63622]&[9,955.4781]&[0.63533]\\
            \midrule
        \multirow{2}{*}{8}&4.92379&1.59602&6.10333&0.10926&6.25755&-0.12557&6.27349&-0.14948\\
            &[8.1786]&[0.93663]&[62.02662]&[0.71209]&[600.749]&[0.69524]&[5,987.81687]&[0.69364]\\
            \midrule
        \multirow{2}{*}{16}&5.34236&2.25233&8.13143&0.84824&8.78207&0.03852&8.85477&-0.05473\\
            &[6.40862]&[1.26388]&[38.17116]&[0.763]&[356.03184]&[0.73054]&[3,534.64538]&[0.72768]\\
            \midrule
        \multirow{2}{*}{32}&5.43021&2.03334&9.25796&1.93719&11.24008&0.38023&11.53585&0.01857\\
            &[5.67156]&[1.9065]&[24.11383]&[0.82876]&[208.54657]&[0.76159]&[2,053.24005]&[0.75641]\\
            \midrule
        \multirow{2}{*}{64}&5.45167&2.00629&9.55211893&2.11940&13.12859&1.25316&14.18264&0.15667\\
            &[5.47496]&[3.26795]&[16.19936]&[0.94472]&[121.5529]&[0.79243]&[1,176.29423]&[0.78272]\\
            \midrule
        \multirow{2}{*}{128}&5.45701&2.00172&9.61982&2.00558&13.92088&2.12106&16.55706&0.61782\\
            &[5.45455]&[2.84155]&[12.08755]&[1.18738]&[71.32525]&[0.82551]&[666.55073]&[0.80617]\\
            \midrule
        \multirow{2}{*}{256}&5.45835&2.00043&9.63668&2.00208&14.10301&2.04849&18.10437&1.63306\\
            &[5.4574]&[1.30897]&[10.28205]&[1.69457]&[42.98257]&[0.86937]&[375.03033]&[0.82732]\\
            \midrule
        \multirow{2}{*}{512}&5.45868&2.00011&9.64089&2.00053&14.14704&2.00316&18.60322&2.15815\\
            &[5.45855]&[2.53479]&[9.72425]&[2.69478]&[27.46815]&[0.94795]&[210.73708]&[0.8473]\\
            \midrule
        \multirow{2}{*}{1024}&5.45876&2.00003&9.64194&2.00013&14.15802&2.00117&18.71499&2.01407\\
            &[5.45875]&[2.49798]&[9.6381]&[6.18789]&[19.42597]&[1.12198]&[119.41889]&[0.86847]\\
            \midrule
        \multirow{2}{*}{2048}&5.45879&2.00001&9.6422&2.00003&14.16077&2.00029&18.74266&2.00253\\
            &[5.45878]&[2.12379]&[9.63928]&[-1.0763]&[15.73088]&[1.51105]&[69.40136]&[0.8969]\\
            \midrule
        \multirow{2}{*}{4096}&5.45879&-&9.64227&-&14.16145&-&18.74956&-\\
            &[5.45879]&-&[9.64177]&-&[14.43443]&-&[42.54005]&-\\
    \bottomrule
    \end{tabularx}
    } % endof \scalebox
    \caption{Results of accuracy and convergence analysis for $\STADD(\mathcal{S}_A)$ for $\theta=1.0$.}
    \label{tab:STADD_vs_A_N__theta__1_00}
\end{table}

%-------------------------------------------------------------------------------------------------%
\section{Conclusion}
\label{sec:conclusion}

We proposed a numerical method to evaluate the performance of the emerging Generalized Shiryaev--Roberts (GSR) procedure in the quickest change-point detection problem's multi-cyclic context. The GSR procedure is an ``umbrella'' term for the original Shiryaev--Roberts procedure and its recent extension---the Shiryaev--Roberts--$r$ procedure. The proposed method is based on the integral-equations approach and uses the collocation framework. To improve the accuracy, robustness and efficiency of the method, the collocation basis functions are selected so as to exploit a certain change-of-measure identity and a certain martingale property of the GSR procedure's detection statistic; efficiency is improved since, by design, the method can compute both the Average Run Length (ARL) to false alarm and the Stationary Average Detection Delay (STADD) {\em simultaneously}. We proved that the method's rate of convergence is quadratic and obtained a tight upperbound on its error. As tested in a case study, the method's expected rate of convergence, greater accuracy and robustness were confirmed experimentally. The method can be used to design the GSR procedure as needed by appropriate selection of the headstart and detection threshold. It is our hope that the proposed method will stimulate further research on as well as application of the GSR procedure in practice.

%-------------------------------------------------------------------------------------------------%
\section*{Acknowledgements}

The authors would like to thank the two anonymous reviewers and Mr.~Yilin Zhu of the Department of Mathematical Sciences at the State University of New York (SUNY) at Binghamton for providing constructive feedback that helped improve the quality of the paper. The authors would also like to thank the Editor-in-Chief, Dr.~Fabrizio Ruggeri of the Institute of Applied Mathematics and Information Technology, Italian National Research Council (CNR IMATI), for the time and effort spent handling the review process to produce this special issue.

An earlier version of this work was presented by A.S.\ Polunchenko at the 56-th Moscow Institute of Physics and Technology Annual Scientific Conference held from November 25, 2013 to November 30, 2013 at said Institute in Moscow, Russia. A.S.\ Polunchenko would like to acknowledge the support provided by the AMS-Simons travel grant to fund his trip to the above Conference.

\bibliographystyle{asmbi}

\bibliography{main,integral-equations,spc,operator-theory,size-bias,numerical-analysis}

\end{document}